\documentclass[12pt]{iopart}
\makeatletter
\expandafter\let\csname equation*\endcsname\relax
\expandafter\let\csname endequation*\endcsname\relax
\makeatother
\usepackage{amsmath}

\usepackage{amsfonts}
\usepackage{amssymb}
\usepackage{graphicx}
\usepackage{physics}
\usepackage{caption}
\usepackage{subcaption}

\usepackage[dvipsnames]{xcolor}
\usepackage{booktabs}
\usepackage{tabularx}
\usepackage{tikz}
\usetikzlibrary{angles,quotes,decorations.markings,decorations.pathreplacing,arrows.meta,patterns,calc}
\usepackage[T1]{fontenc}
\usepackage[english]{babel}
\usepackage[expansion=false]{microtype}
\usepackage[hyphens]{url}
\usepackage{datetime}
\usepackage[normalem]{ulem}

\newcommand{\E}{\mathbb{E}}
\newcommand{\Var}{\operatorname{Var}}

\newcommand{\gstar}{g^{*}}

\usepackage{amsthm,url}
\usepackage{hyperref}
\hypersetup{
    colorlinks=true,
    linkcolor=MidnightBlue,
    citecolor=ForestGreen,
    urlcolor=BrickRed,
    pdfauthor={},
    pdftitle={}
}

\newtheorem{proposition}{Proposition}

\newtheorem{corollary}{Corollary}
\newtheorem{assumption}{Assumption}
\newtheorem{remark}{Remark}

\begin{document}

\title[Measurement cost of gradient attacks on QML]{When cheap gradients fail: the measurement cost of attacking quantum classifiers}

\author{Bacui Li$^{1,2}$, Chandra Thapa$^{2}$, Tansu Alpcan$^{1}$ and Udaya Parampalli$^{3}$}

\address{$^{1}$ Department of Electrical and Electronic Engineering, University of Melbourne, Parkville, Victoria 3010, Australia}
\address{$^{2}$ CSIRO Data61, Marsfield, NSW 2122, Australia}
\address{$^{3}$ School of Computing and Information Systems, University of Melbourne, Parkville, Victoria 3010, Australia}

\ead{bacuil@student.unimelb.edu.au}

\begin{abstract}
Adversarial perturbations threaten machine learning classifiers, including variational quantum classifiers.
We show that finite quantum measurement statistics, that is, shot noise, act as a built-in defense against gradient-based test-time attacks whose cost scales unfavorably for the attacker.
Because every gradient component must be inferred from repeated circuit executions under any unbiased gradient-estimation rule, white-box extraction consumes a dimension-dependent measurement budget that measurement grouping cannot remove in expressive circuits.
We establish, under stated assumptions, that single-step attacks need at least quadratically many shots in the input dimension $d$, growing as $d^{5/2}$ under norm-concentration scaling, with a sufficient-budget analysis for iterative attacks via stochastic gradient Langevin dynamics.
Simulations up to 784 input dimensions validate the law: the realized total budget is the $d^{5/2}$ geometric floor for plateau-mitigated models and grows as $d^{3.00}$ for the tested deep circuits, whose gradient norms decay with dimension absent barren-plateau mitigation; folding the measured gradient norm back in per sample recovers the parameter-free $d^{3/2}$ shot-noise geometry (measured $d^{1.46}$).
Against a matched classical baseline whose attack overhead is dimension-independent (the cheap-gradient principle of automatic differentiation), the quantum \emph{gradient cost ratio}, a gradient's cost in forward-inference units, grows polynomially, empirically as $d^{3.00}$, so the attacker's relative cost diverges as the model scales.
On a 156-qubit IBM processor (\texttt{ibm\_boston}, 4-qubit circuits, $d{=}12$), a simulator--hardware comparison over a 100-input cohort reproduces the effect, the device attack tracking the ideal within a few percent at matched budgets, with the high-shot gradient faithful to the exact one (cohort-median cosine $0.98$, mean $0.90$).
The experiments establish the \emph{scaling law} of measurement-based gradient extraction; its defensive consequence operates precisely when the forward map is classically hard to simulate, since only then is a white-box attacker denied the simulate-and-backpropagate shortcut and must pay the measurement cost we quantify.
\end{abstract}

\vspace{1pc}

\vspace{1pc}

\section{Introduction}

Modern machine learning models achieve high accuracy under benign conditions yet remain vulnerable to adversarial perturbations: small, deliberately crafted input changes that flip a classifier's prediction while remaining imperceptible to a human observer~\cite{goodfellow_explaining_2014, carlini_towards_2016, bruna_intriguing_2013}. Quantum machine learning (QML)~\cite{melnikov_quantum_2023}, whose models are parametrized quantum circuits trained on classical or quantum data, inherits this vulnerability. Existing defenses for QML are largely \emph{dynamical}: worst-case robustness degrades only polynomially in qubit count~\cite{liu_vulnerability_2020, liao_robust_2021}, a bound that data-geometry refinements sharpen by exploiting the manifold structure of natural data~\cite{mahloujifar_curse_2018} and that has been benchmarked on quantum classifiers at scale~\cite{west_benchmarking_2023}; alongside these, empirical studies report that depolarization, hardware crosstalk, and label noise each enhance robustness on near-term devices~\cite{du_quantum_2021, kundu_qnad_2024, ahmed_comparative_2025, zhang_experimental_2025}. Dowling et al.~\cite{dowling_adversarial_2024} organize these protections into a hierarchy of dynamical guarantees rooted in unitarity, operator scrambling, and circuit chaoticity. Every analysis in this line, however, prices the attacker's gradient at zero: its guarantees are stated for exact expectation values.

Quantum mechanics forbids the zero-cost gradient. Every observable must be reconstructed from finitely many measurements, and the parameter-shift rule (PSR)~\cite{mitarai_quantum_2018, schuld_evaluating_2019} makes the cost explicit: each gradient component is a difference of two expectation values, each estimated from shifted-circuit measurement statistics. Such estimation is governed by the quantum Cram\'er--Rao bound~\cite{helstrom_quantum_1969, braunstein_statistical_1994}, which for single-copy projective measurements coincides with the classical shot-noise variance and floors the variance of any unbiased estimator built from a finite number of shots; under PSR this floor is a per-component variance $\sigma^2/s$ for $s$ shots. A substantial training-time literature amortizes this cost across epochs through adaptive shot-allocation schedules~\cite{scriva2023shotnoise, kubler_adaptive_2020, bittel_fast_2022, ito_santaqlaus_2023, kreplin_reduction_2024, ma_adaptive_2021}; the security side has none. The asymmetry is the point: a defender pays the shot bill once, spread across training, whereas an attacker pays it afresh for every gradient query at inference time. This sets quantum shot noise apart from the defenses a model designer adds deliberately. Algorithmic gradient obfuscation offers only a false sense of security, since adaptive attacks defeat it at negligible additional cost~\cite{athalye_obfuscated_2018}; certified randomized-smoothing schemes~\cite{cohen2019smoothing, wong2018convexouter} do supply provable robustness radii, but at the price of test-time sampling overhead and radii that themselves shrink with dimension. Quantum shot noise is intrinsic to measurement: it cannot be switched off without buying more shots, and it acts on every gradient component of every query the attacker issues.

Our analysis complements the dynamical-guarantee literature~\cite{dowling_adversarial_2024}: those results bound attack feasibility through circuit properties, whereas we quantify the physical measurement cost that any gradient-based attack incurs regardless of the circuit's dynamical regime. It also recasts the cheap-gradient theorem~\cite{baur_complexity_1983, griewank_evaluating_2008}. Classically, the \emph{gradient cost ratio}, the cost of a gradient in units of one forward pass, is bounded by a small constant ($\le5$); on quantum hardware it is exactly this ratio that blows up, because unbiased gradient estimation costs $\Theta(d)$ measurements, linear in the input data dimension $d$, regardless of the algorithm used. This $\Theta(d)$ component cost cannot be reduced by any measurement grouping. By the expressivity--measurement-efficiency trade-off for parametrized quantum circuits~\cite{chinzei_tradeoff_2025}, the number of simultaneously measurable gradient components collapses to $O(1)$ for deep, generically-entangling ans\"atze whose dynamical Lie algebra (the Lie algebra generated by the circuit's gate Hamiltonians) is near-maximal, that is, fills the full $\mathfrak{su}(2^q)$, which is also the regime in which the forward map becomes classically hard; the per-gradient measurement count is then proportional to the number of components. The exact parameter-shift family is unbiased and so falls under this $\Theta(d)$ bound: generalized parameter-shift rules for gates with richer spectra~\cite{Wierichs2022, banchi_overshifted_2025} change only the number of shifts per component and remain unbiased, so they do not lower the exponent. Biased or stochastic surrogates instead sit outside the unbiased bound: simultaneous-perturbation (SPSA) estimators and PSR/SPSA hybrids~\cite{periyasamy_guided-spsa_2024, hoffmann_gradient_2022} and classical-shadow gradients~\cite{heidari_quantum_2024} reduce the per-component circuit count at the price of estimator variance and, for finite-difference and SPSA-type rules, bias; whether that trade lowers the exponent in the high-curvature expressive regime is governed by the model's curvature and remains open (\ref{app:zero-order}; see also the empirical survey of~\cite{lockwood_empirical_2022}).

\paragraph{Setting and threat model.} We make the attacker precise before stating our results. We study a test-time, white-box \emph{evasion} attack~\cite{biggio_evasion_2013} on a trained QML classifier of classical data. A classifier $f(x,\theta)$ maps a classical input $x\in\mathbb{R}^{d}$ to class scores through a parametrized quantum circuit. Training is complete and the parameters $\theta$ are frozen; at inference the adversary adds a perturbation $\delta$ with $\|\delta\|_2\le\epsilon$ (the perturbation budget $\epsilon$) to a single input so that $f(x+\delta,\theta)$ is misclassified. A white-box adversary knows the architecture and $\theta$; this is the strongest such attacker, so any cost it must pay, a black-box attacker pays as well. The attack direction is set by the input gradient $\gstar=\nabla_x L(x,y)$ of the loss $L$. On classical hardware this gradient is essentially free to compute, but on a quantum device every component of $\gstar$ must be reconstructed from finitely many circuit measurements. We write $R$ for the total measurement budget (circuit executions, or shots) the attacker spends to do so. The question this paper answers is how the budget $R$ must grow with $d$ to keep the attack effective.

\paragraph{Contributions.} We answer this in four parts. (i)~\emph{Single-step attacks.} Under stated, attacker-favoring assumptions (local linearity, i.i.d.\ shot noise, and a small-noise regime), the single-step shot cost is $R=\Theta(\epsilon d^{2})$: $\Theta(d^{2})$ at a fixed perturbation budget and $\Theta(d^{5/2})$ under the $\ell_2$ scaling $\epsilon\propto\sqrt{d}$ that norm concentration motivates; because each assumption favors the attacker, relaxing any of them only raises the cost. (ii)~\emph{Iterative attacks.} Casting the finite-shot Carlini--Wagner attack~\cite{carlini_towards_2016} as Stochastic Gradient Langevin Dynamics (SGLD) with inverse temperature $\beta\propto s$~\cite{raginsky_non-convex_2017, xu_global_2020} reproduces the empirical shot sweet spot and shows that finite-shot convergence theory grants an iterative attacker no guarantee of escaping the single-step scaling (\ref{app:sgld_sufficient}). (iii)~\emph{Simulation.} Simulations up to $d{=}784$ on the MNIST and Fashion-MNIST image datasets confirm the law's structure: folding the measured gradient norm back into the cost coefficient recovers the parameter-free $d^{3/2}$ geometric baseline, and for the tested \emph{plateau-prone} circuits (gradient norm decaying with dimension, $\|\gstar\|\propto d^{-0.73}$) the realized total budget grows approximately cubically ($d^{3.00}$ and $d^{3.07}$ on the two datasets), above the $d^{5/2}$ floor that \emph{plateau-mitigated} models (dimension-stable gradients, $\|\gstar\|=\Theta(1)$, e.g.\ barren-plateau-mitigated architectures) would reach. A matched classical CNN baseline has a dimension-independent gradient cost ratio $\rho_\mathrm{classical}\approx5$ (the Baur--Strassen bound), so the quantum overhead's growth as $d^{3.00}$ makes the relative gradient-extraction cost diverge polynomially with $d$. (iv)~\emph{Hardware.} On IBM's 156-qubit \texttt{ibm\_boston}, a matched simulation-versus-experiment comparison at $d=12$ over a 100-input cohort reproduces the shot-noise robustness; the high-shot gradient is faithful (cosine $0.90$). The simulation reaches the exact-gradient floor ($10\%$) as $s\to\infty$, and the experiment plateaus a few percent above it, a residual device bias the attacker cannot remove, so the low-shot robustness reflects finite measurement and survives device imperfection.

\label{sec:scope}
Together, these results show that shot noise induces a dimension-scaling resource asymmetry for measurement-based gradient extraction. The scaling itself is generic to noisy, gradient-free (zero-order) estimation: recovering a $d$-dimensional gradient from such an oracle is long known to cost $\Theta(d)$ queries, with matching information-theoretic lower bounds~\cite{nesterov_random_2017, jamieson_query_2012, duchi_optimal_2015}. What is specific to the quantum setting is the \emph{irreducibility} of that cost under our threat model. The per-query noise is the Born rule under single-copy measurement, intrinsic to the physics rather than a software perturbation that a white-box attacker could switch off or cheaply average away~\cite{athalye_obfuscated_2018}; and in the classically-hard regime, quantum hardware denies the simulate-and-backpropagate shortcut even to an attacker who knows $\theta$, leaving an exponent that no \emph{unbiased} estimator can reduce, because the gradient components cannot be co-measured~\cite{chinzei_tradeoff_2025}. The shot-budget bounds concern inference-time gradient extraction only. We further assume the attacker processes one copy of the inference circuit at a time, using the same number of qubits and exploiting no cross-circuit entanglement, so that each shot resolves a single circuit execution rather than a joint multi-copy measurement. Training-time threats such as data poisoning and backdoor injection lie outside this framework and are formalized separately in Section~\ref{sec:framework}. Table~\ref{tab:roadmap} summarizes the four attack regimes considered, the assumptions and tools used to bound each, and the section in which each result appears.

\paragraph{Regime of applicability.} One scoping point is essential to reading these results correctly. The defensive consequence of shot noise operates precisely when classical simulation of the model's forward map is intractable. A white-box attacker who can simulate the circuit can also backpropagate through that simulation to obtain the input gradient at $O(1)$ cost (Remark~\ref{rem:classical_cost}), bypassing the measurement bill entirely. The setting in which QML is expected to be useful, where its forward map offers a genuine quantum advantage, is by definition one in which that map is classically hard to reproduce. There, and only there, the attacker is denied the simulate-and-backpropagate shortcut and must extract every gradient component through measurement, paying the dimension-dependent cost we bound. The small systems we simulate ($q\le 10$ qubits) lie deliberately on the tractable side of this boundary, so that we can compute ground-truth gradients and isolate the scaling exponent; at those sizes a rational white-box attacker would simulate rather than measure. Our experiments should therefore be read as establishing the \emph{scaling law} of measurement-based gradient extraction, not as a deployed defense at the dimensions probed. Current hardware and the cost of classical simulation place the classically-hard regime beyond what a single study can instantiate end to end, but the scaling law is what characterizes the attacker's cost once that boundary is crossed. What transfers across the boundary is the \emph{exponent}: the per-component count is fixed by the input dimension and, by the expressivity--measurement-efficiency trade-off~\cite{chinzei_tradeoff_2025}, cannot be amortized by any measurement grouping, precisely because the useful regime is the expressive one. The ans\"atze we simulate already lie in this expressive class (deep, generically-entangling strongly-entangling-layer circuits whose dynamical Lie algebra is near-maximal), so the no-amortization premise is expected to hold for our circuit family at every tested size; only the classical hardness of the forward map awaits larger $q$, not the expressivity that fixes the exponent. The \emph{prefactor} carries the per-observable measurement variance $\sigma^2$, which the strong entanglement and output concentration of the expressive regime may shift; bounding it there is the boundary of our extrapolation, which we leave to future work.

\begin{table}
\centering
\small
\caption{\label{tab:roadmap}Roadmap of the four attack regimes analyzed. $R$ denotes the total shot budget required to maintain constant attack efficacy at input dimension $d$. The empirical total-shot exponent (denoted $p_1$; fitting convention in \ref{methods:software}) exceeds the $d^{5/2}$ baseline because the tested models' gradient norms decay with dimension, violating the baseline's $\|\gstar\|=\Theta(1)$ premise in the defender's favor (Section~\ref{sec:Res}). C\&W is studied as a representative example of the broader, largely heuristic, multi-step attack landscape in QML.}
\begin{tabularx}{\linewidth}{@{}lXlc@{}}
\toprule
\textbf{Regime} & \textbf{Assumption / tool} & \textbf{Result} & \textbf{Section} \\
\midrule
Single-step, fixed $\epsilon$ & local linearity, i.i.d.\ shots & $R = \Theta(\epsilon d^{2})=\Theta(d^{2})$ & \ref{sec:singlestep} \\
Single-step, $\epsilon \propto \sqrt{d}$ & norm concentration & $R = \Theta(\epsilon d^{2})=\Theta(d^{5/2})$ & \ref{sec:singlestep} \\
Iterative (C\&W exemplar) & SGLD + Raginsky~\cite{raginsky_non-convex_2017} & sufficient poly($d$) budget & \ref{sec:multistep} \\
Empirical \& hardware & matched CNN baseline; \texttt{ibm\_boston} & $p_1 \approx 3.0$; overhead $\Theta(d^{3.00})$ vs $O(1)$ & \ref{sec:Res}, \ref{sec:hw_validation} \\
\bottomrule
\end{tabularx}
\end{table}

\section{Theoretical framework and threat model}
\label{sec:methods}
\label{sec:background}

This section establishes the measurement framework and shows that shot noise induces dimension-dependent gradient uncertainty, the foundation for understanding what happens when a quantum classifier is attacked with finite resources. We introduce notation and the PSR gradient-estimation pipeline including PSR unbiasedness and covariance bookkeeping (Section~\ref{sec:notation}), formalize the threat model and the single-step attack assumptions (Section~\ref{sec:framework}), derive the chain-rule structure of noisy gradient estimates (Section~\ref{sec:uncertain_gradient}), and establish forward- versus backward-pass cost dominance (Section~\ref{methods:fp_bp_dominance}).

\subsection{Notation and gradient estimation via PSR}
\label{sec:notation}

We consider a quantum machine learning classifier $f(x,\theta): \mathbb{R}^d \to \mathbb{R}^C$ mapping $d$-dimensional inputs to $C$ class logits via a parametric quantum circuit with trainable parameters $\theta$; Table~\ref{tab:notation} collects the symbols used throughout this section and the rest of the paper. The quantum layer employs angle encoding gates $U(x_i) = e^{iG x_i}$ where $G$ is a generator with two distinct eigenvalues, producing observables whose expectation values $\langle O_j \rangle$ define output features. The $C$ class logits are these expectation values directly: the model carries no trainable classical readout layer, and the only classical post-processing is the softmax cross-entropy loss, whose gradient with respect to the logits costs $O(1)$ by automatic differentiation and so leaves the dimension-dependent measurement cost of the quantum Jacobian unchanged. This setup represents the minimal structure required for parameter-shift gradient estimation and is general: it encompasses widely-used QML architectures, including data re-uploading circuits~\cite{perez-salinas_data_2020}, Instantaneous Quantum Polynomial (IQP) models~\cite{havlicek_supervised_2019}, and variational quantum classifiers~\cite{schuld_circuit-centric_2020}. Any parametric gate satisfying Lie-algebraic conditions for the parameter-shift rule~\cite{schuld_evaluating_2019} falls within this framework, making our shot-noise analysis broadly applicable to practical QML implementations.

Gradient estimation via the Parameter-Shift Rule (PSR) exploits the spectral structure of encoding gates. For a gate $U(x_i) = e^{iG_i x_i}$ whose generator $G_i$ has two distinct eigenvalues with spectral gap $\Delta G_i = \lambda_+ - \lambda_-$, the exact gradient is
\begin{equation}
\label{eq:psr_intro}
\begin{aligned}
\frac{\partial\langle O\rangle}{\partial x_i}
&= \frac{\Delta G_i}{2}\!\left[\langle O(x^{(i,+)})\rangle - \langle O(x^{(i,-)})\rangle\right],\\
\quad x^{(i,\pm)} &= x \pm \frac{\pi}{2\Delta G_i}\,\hat{e}_i\;,
\end{aligned}
\end{equation}
where $\hat{e}_i$ is the unit vector along input coordinate $i$ (so $x^{(i,\pm)}$ shifts only the $i$-th feature). For single-qubit Pauli rotations ($G_i = \sigma_i/2$, $\Delta G_i = 1$), this reduces to the familiar shift $a = \pi/2$. Each shifted expectation estimated from $s_0$ shots has variance $\mathrm{Var}[\hat{O}_j] = (1 - \langle O_j\rangle^2)/s_0 \leq 1/s_0$ for $\pm 1$-valued Pauli observables, where $\hat{O}_j$ is the $s_0$-shot empirical estimate of $\langle O_j\rangle$. The gradient estimate inherits this uncertainty: $\mathrm{Var}[\partial\hat{O}_j/\partial x_i] = \Theta(\Delta G_i^2/s_0)$. Estimating the full Jacobian $J\in\mathbb{R}^{C\times d}$ requires $2C$ circuit evaluations per input dimension (two PSR shifts for each of $C$ output observables), yielding $2Cd$ evaluations in total, each using $s_0$ shots. When all $C$ output observables commute (e.g., computational-basis $Z_i$ measurements, as in our experiments), the $C$ expectation values share the same bitstrings, reducing the count to $2$ evaluations per input dimension (Section~\ref{methods:fp_bp_dominance}).

\textbf{Convention.} Throughout the paper, $s$ denotes the total shots \emph{per input dimension}, absorbing all per-evaluation overheads: $s = 2C\,s_0$ in the non-commuting case, or $s = 2\,s_0$ when observables commute. The total shot budget is then uniformly
\[
  R \;=\; d\,s\;.
\]
The per-evaluation count $s_0$ appears only in this derivation; all subsequent scaling laws, propositions, and experiments are stated in terms of $s$ and $R = ds$. The effective per-shot variance $\sigma^2$ used in later sections absorbs the constant relating $s$ to $s_0$, so that the gradient noise model takes the simple form $\mathrm{Var}[\hat{g}_i]=\sigma^2/s$.
This measurement uncertainty propagates through backpropagation, inducing gradient noise that is the central object of our analysis.
We measure perturbation size via $\ell_p$ norms $\|\delta\|_p$ and quantify robustness primarily through \emph{post--attack accuracy} under an $\epsilon$-bounded perturbation (equivalently, attack success rate), together with loss-based surrogates such as the loss surplus $\Delta L(s)$ used in our scaling bounds. Our analysis focuses exclusively on the $\ell_2$ norm for three reasons. \emph{(i) Shot-noise geometry:} each gradient component estimated via PSR has independent Gaussian noise with variance $\sigma^2/s$, so the total estimation error $\|\hat{g} - \gstar\|$ concentrates in $\ell_2$, making $\ell_2$-norm constraints the natural measure of gradient fidelity. \emph{(ii) Optimal attack alignment:} under $\ell_2$-bounded perturbations, the optimal first-order attack direction is $\delta^* = \epsilon \gstar/\|\gstar\|_2$, aligning exactly with the normalized gradient; for other $\ell_p$ norms ($p \neq 2$), the dual-norm structure yields optimal perturbations that do not align with the gradient direction (e.g.\ the $\ell_\infty$ sign-gradient attack~\cite{goodfellow_explaining_2014, madry_towards_2017}), complicating the analysis of how shot noise degrades attack effectiveness. \emph{(iii) Practical comparability:} in classical adversarial robustness, different $\ell_p$ norms produce correlated but distinct vulnerability measures (an attack effective under $\ell_\infty$ constraints may perform differently under $\ell_2$ bounds), and the $\ell_2$ norm provides a natural baseline for QML that directly reflects the Euclidean structure of gradient estimation noise.
The attacker's resource cost $R$ denotes the total number of circuit executions (shots $\times$ evaluations) required to craft adversarial examples.

\begin{table}[t]
\centering
\small
\begin{tabularx}{\linewidth}{cX}
\hline
\textbf{Symbol} & \textbf{Meaning} \\
\hline
$d$ & Input dimension (number of features) \\
$d_\textup{in}$ & Input dimension to a specific quantum layer (layer local)\\
$s$ & Shots per input dimension (total across all PSR circuit evaluations for one coordinate; see Section~\ref{sec:notation}) \\
$R$ & Total shot budget $R = ds$ \\
$\epsilon$ & Perturbation budget ($\ell_2$-norm bound) \\
$\gstar$ & True gradient vector $\nabla_x L(x,y)$ \\
$\hat{g}$ & Noisy gradient estimate from finite shots \\
$\sigma^2$ & Per-shot variance of observable measurements \\
$\Sigma$ & Covariance matrix of gradient estimation error \\
$J_k$ & Jacobian of quantum layer $k$ w.r.t. inputs \\
$L(x,y)$ & Attacker's loss function (rewards misclassification) \\
$L_{\text{train}}$ & Defender's training loss \\
$\alpha$ & Misalignment angle between $\hat{g}$ and $\gstar$ \\
$\kappa_{\mathrm{vMF}}$ & vMF concentration parameter: $\kappa_{\mathrm{vMF}} \approx s\|\gstar\|^2/\sigma^2$ \\
$\zeta(d)$ & Landscape factor: measured ratio of the realized clean-attack loss rise to its first-order prediction (\S\ref{par:gnorm_baseline}); $\zeta(d)=1$ for a locally linear landscape \\
$\Delta G_i$ & Spectral gap of generator $G_i$ ($=\lambda_+-\lambda_-$)\\
$C$          & Number of output classes\\
$\beta_\mathrm{eff}$ & SGLD effective inverse temperature\\
\hline
\end{tabularx}
\caption{\textbf{Notation.} Symbols used throughout the main text.}
\label{tab:notation}
\end{table}

\subsection{Threat model: gradient-based attacks \label{sec:framework}}
We focus on physically realistic \emph{white-box, test-time (evasion) attacks} on QML classifiers~\cite{biggio_evasion_2013}: the classifier $f(\cdot,\theta)$ has already been trained, its parameters $\theta$ are frozen, and the adversary's only freedom is to craft a small input perturbation $\delta$ at inference (query) time so that $f(x+\delta,\theta)$ is misclassified. Training-time adversaries (any attack that modifies $\theta$ or the training distribution) are explicitly out of scope; the resource bounds in this paper concern only the inference-time, per-sample cost of gradient extraction. Within this evasion setting, a \textit{white-box attack} further assumes complete knowledge of the model architecture and trained parameters $\theta$, enabling the adversary to compute or estimate gradients $\nabla_x L(x,y)$ for crafting perturbations (the white-box assumption still requires physical gradient extraction via quantum measurements). This represents a worst-case threat stronger than black-box scenarios.\footnote{We assume PSR-based gradient estimation; see Section~\ref{sec:scope} and \ref{app:zero-order} for discussion of alternative methods.}

\paragraph{Beyond gradient-based attacks.} The measurement cost we bound is a tax on \emph{gradient reconstruction}; it does not directly bind attacks that avoid estimating $\nabla_x L$, namely gradient-free search and transfer (surrogate) attacks, which constitute a separate threat model. Transfer is an imperfect route against quantum classifiers: variational quantum models exhibit enhanced adversarial robustness and learn features distinct from those of classical networks~\cite{west_benchmarking_2023}, so adversarial examples crafted on a classical surrogate transfer only partially and asymmetrically, though transfer is not eliminated~\cite{lu_quantum_2020}. A full treatment of gradient-free and transfer attacks under the finite-shot model is left to future work.

\paragraph{Classification setup.}
Input samples lie in a bounded space $x \in [0,1]^d$ of dimension $d$ with labels $y \in \{1,\ldots,C\}$. The QML classifier $f(x, \theta): \mathbb{R}^d \to \mathbb{R}^C$ outputs logits (unnormalized scores) for each class; the predicted label is $\hat{y}=\arg \max_c f(x, \theta)_c$. Training minimizes a differentiable loss $L_\textup{train}(x,y;\theta)$, typically cross-entropy after softmax normalization, though alternatives (hinge, focal) exist.

\paragraph{Attack objective.}
An adversary seeks a perturbation $\delta$ that causes misclassification while remaining small to evade detection. This trade-off is formalized as a multi-objective optimization:
\begin{align}
\label{eq:multi_obj}
    \min_{\delta}\;
    &  L(x+\delta, y ; \theta)
       \;\text{and}\; \|\delta\|_p, \notag\\
    \text{s.t.}\;
    &  x+\delta\in[0,1]^d,
\end{align}
where $L$ is an attack loss rewarding misclassification (e.g., negative cross-entropy, difference-of-logits variants) and $\|\delta\|_p$ measures perturbation size under the $\ell_p$ norm. The constraint ensures adversarial examples remain in the input domain.

\paragraph{Attack families and the $\ell_2$ FGSM.}
Attacks trade off the two objectives differently: fixed-budget methods (FGSM, PGD) fix $\|\delta\|_p=\epsilon$ and minimize $L$, while minimum-norm and mixed methods (C\&W) minimize $\|\delta\|_p$ or a weighted combination $L+c\|\delta\|_p^2$. We use the $\ell_2$-normalized single-step attack
\begin{equation}
\label{eq:fgm}
    \delta = \epsilon \cdot \frac{\nabla_x L}{\|\nabla_x L\|_2}\;,
\end{equation}
which is optimal under $\ell_2$ constraints~\cite{madry_towards_2017} and preserves the relative magnitudes of all gradient components, the property our shot-noise analysis depends on; it is the $\ell_2$ analog of the sign-based FGSM~\cite{goodfellow_explaining_2014}. Throughout, ``FGSM'' denotes this $\ell_2$ variant.

\paragraph{Attack effectiveness metrics.}
Throughout this paper, we evaluate attacks at a prescribed perturbation budget $\epsilon$ and report \emph{post--attack accuracy} (equivalently, attack success rate) as the primary empirical outcome. On the theory side, we also track loss-based surrogates, in particular the loss surplus $\Delta L(s)=L(\delta)-L(\delta^*)$ that appears in Proposition~\ref{prop:fgsm-optimization}, because it connects directly to gradient misalignment under shot noise. For iterative attacks such as C\&W, we report the same metrics after a fixed number of optimization steps under a fixed total shot budget. We do not perform per-sample searches for minimal perturbations, and we do not include hyperparameter tuning overhead (e.g., binary search over the C\&W coefficient $c$) in the resource count, so accounting for this overhead would only strengthen the resource lower bounds.

\paragraph{Single-step attacks and working assumptions.}
A single-step gradient-based attack with perturbation budget $\epsilon$ takes the form
\begin{equation}
\label{eq:1s_attack}
    \delta = \epsilon\, F\!\left[\nabla L(x, y)\right],
\end{equation}
where $F(\cdot):\mathbb{R}^n\to\mathbb{R}^n$ maps the gradient to a perturbation (e.g.\ $F=\mathrm{sign}$ for the original FGSM, $F(\gstar)=\gstar/\|\gstar\|_2$ for Eq.~\eqref{eq:fgm}). The single-step shot-scaling analysis in Section~\ref{sec:singlestep} works under three assumptions, all favorable to the attacker:
\begin{assumption}[Linear local landscape]
\label{ass:1s_1}
The loss landscape is locally linear around $x$ within the perturbation ball, i.e.\ for any $\delta$ with $\|\delta\|_p\le\epsilon$, $L(x+\delta,y)=L(x,y)+\nabla L(x,y)\cdot\delta$.
\end{assumption}
\begin{assumption}[i.i.d.\ Gaussian gradient components]
\label{ass:1s_2}
The estimated loss-gradient components are independent normal, $\hat{g}_i\sim\mathcal N(\gstar_i,\sigma^2/s)$, where $s$ is the average per-dimension shot count and $\sigma^2$ is the effective per-shot variance; the total shot budget is $R=ds$.
\end{assumption}
\begin{assumption}[Small noise]
\label{ass:1s_3}
The per-component noise is much smaller than the true gradient magnitude, $\sqrt{\sigma^2/s}\ll\|\gstar\|_2$, equivalently $s\gg s_{\mathrm{vMF}}(d):=\sigma^2/\|\gstar\|^2$.
\end{assumption}
A relaxation that drops Assumptions~\ref{ass:1s_2}--\ref{ass:1s_3} and only requires zero-mean gradient error with $\tr\Sigma=\Theta(d)$ is given in Section~\ref{subsec:expected_relax}.

\subsection{The uncertainty of gradient estimates \label{sec:uncertain_gradient}}
In classical machine learning, gradient-based adversarial attacks obtain the required input gradients through automatic differentiation and backpropagation. Attackers therefore receive exact, low-variance gradients at a small constant-factor cost relative to inference (the cheap-gradient principle, $\le5{\times}$; Remark~\ref{rem:classical_cost}) \cite{baydin2018autodiff,goodfellow_explaining_2014}. Quantum models break this convenience through finite‑shot measurement uncertainty: every gradient component must be inferred from measurement statistics, so stochastic estimation becomes unavoidable even in a white-box regime, introducing irreducible uncertainty into adversarial gradient extraction. In this section, we formalize the gradient-estimation pipeline for QML attackers and quantify how shot noise disrupts the perturbations they attempt to construct.

During a white-box attack, the adversary knows the architecture and trained parameters. Let $L(x,y)$ denote the attacker's loss for input-label pair $(x,y)$. It differs from the training objective, but we drop any explicit model subscripts for clarity. The attacker seeks $g = \nabla_x L(x,y)$, the gradient with respect to the input features. When a model contains at least one quantum subroutine whose forward map is classically intractable, the adversary cannot simulate it and must execute that subroutine on quantum hardware to estimate $g$, since classical simulation becomes infeasible for generic large-scale circuits \cite{feynman_simulating_1982, bernstein_quantum_1997}. This is exactly the regime in which the shot-cost bounds below carry defensive force (Section~\ref{sec:scope}): where the circuit is small enough to simulate, a white-box attacker would instead backpropagate through the simulation at $O(1)$ cost, and the bounds describe the cost of measurement-based extraction rather than an operative barrier.

Write the model as a composition of $M$ layers $f_m$ so that $f(x) = f_M \circ f_{M-1} \circ \cdots \circ f_0(x)$. Differentiating the loss with respect to $x$ yields the multivariable chain-rule factorization
\begin{equation}
\label{eq:chain}
    g = J_M(x^{(M)}) J_{M-1}(x^{(M-1)})...J_0(x),
\end{equation}
where $J_m$ is the Jacobian of the $m$th layer.

When a quantum block participates in this composition, measurement randomness injects uncertainty into its Jacobian estimate, and that uncertainty propagates through any surrounding classical layers. For analytic tractability, we discuss a single quantum layer at position $k$ with Jacobian $J_k$. Still, the same reasoning applies to architectures with multiple quantum segments: each stochastic Jacobian multiplies the deterministic Jacobians of adjacent classical modules.

Define the effective Jacobian before the quantum layer as $J_L = J_{k-1}J_{k-2}\cdots J_0$ and the effective Jacobian after it as $J_U = J_M J_{M-1}\cdots J_{k+1}$. Substituting into (\ref{eq:chain}) gives
\begin{equation}
\label{eq:g_jac}
    \hat{g} = J_U \hat{J_k} J_L\;.
\end{equation}
Here, we also replace the analytical gradient with an estimate of the gradient and match it to $\hat{J_k}$, the estimate of the Jacobian quantum layer on the RHS. Note that $J_U$ has an implicit dependence on the output of the quantum layer $f_k$; we treat $J_U$ as known in the propositions below because the forward pass of $f_k$ is much cheaper than its backward pass (Section~\ref{methods:theory}, Proposition~\ref{prop:backward_dominance}).

\paragraph{Statistical properties of the gradient estimate.}
PSR is unbiased, $\mathbb E[\hat J_k]=J_k$, and each Jacobian entry obtained by the two-shift differencing of Eq.~\eqref{eq:psr_intro} has variance scaling
\[
\mathbb V[\hat J_{ij}] = \Theta\!\left(\frac{\Delta_{G_j}^2}{s_j}\right),
\]
where $s_j$ is the shots allocated to parameter $j$ and the constant depends on the observable statistics and shot-allocation scheme~\cite{vanDerVaart1998}. Commuting outputs (e.g.\ simultaneous $\{Z_i\}$ readout) share bitstrings, correlating the rows within a Jacobian column. Non-commuting outputs measured in separate settings give independent entries when shot pools do not overlap~\cite{Yen2023,Crawford2021}. The effective gradient $\hat G=J_U\,\hat J\,J_L$ then has mean $J_U\,\mu\,J_L$ with $\mu=\mathbb E[\hat J]$ and column-wise covariance
\begin{equation}
\label{eq:variance1}
\mathrm{Cov}(\hat G_{\cdot i},\hat G_{\cdot j})=\sum_{v=1}^{d}\!\big(J_U\,\Sigma_v\,J_U^\top\big)\,[J_L]_{v i}[J_L]_{v j},\qquad \Sigma_v=\mathrm{Cov}(\hat J_{\cdot v}).
\end{equation}
The per-case covariance bookkeeping and the $\mathrm{vec}$-level Kronecker form are collected in \ref{app:supp_derivations}. Consistent with PSR error analyses, the per-component gradient-estimation error scales as $s^{-1/2}$ (equivalently $R^{-1/2}$ for total budget $R=ds$)~\cite{Mari2021,Kaminishi2024}.

\subsection{Forward- versus backward-pass shot cost}
\label{methods:theory}
\label{methods:fp_bp_dominance}
This subsection establishes that, for a hybrid architecture with a quantum layer $f_k:\mathbb{R}^{d_{\mathrm{in}}}\to\mathbb{R}^{d_{\mathrm{out}}}$ surrounded by classical layers, backward-pass Jacobian estimation dominates the shot budget in high dimensions. Accounting for both forward-pass (output) and backward-pass (Jacobian) measurement noise, the gradient estimate of Eq.~\eqref{eq:g_jac} becomes $\hat{g} = J_U(\hat{y}_k)\,\hat{J}_k\,J_L$, where $\hat{y}_k$ is the noisy quantum output. Proposition~\ref{prop:forward_variance} characterizes the forward contribution; Proposition~\ref{prop:backward_dominance} and Corollary~\ref{cor:backward_focus} establish the dominance of the backward cost.

\begin{assumption}[Smoothness of classical post-processing]
\label{ass:fp_noise}
The classical layers following the quantum block are twice continuously differentiable with bounded Hessian in a neighborhood of $y_k$, so that the Jacobian satisfies $J_U(\hat{y}_k) - J_U(y_k) = H_U(\hat{y}_k - y_k) + O(\|\hat{y}_k - y_k\|^2)$ for a bounded matrix $H_U$.
\end{assumption}

\begin{proposition}[Forward-pass variance scaling]
\label{prop:forward_variance}
Under Assumption~\ref{ass:fp_noise}, the variance contribution from forward-pass noise satisfies $\mathbb{V}[J_U(\hat{y}_k) - J_U(y_k)] = \Theta(1/m)$, where $m$ is the number of shots per observable and the constant depends on $\|H_U\|$ and the covariance structure given in equation~\eqref{eq:variance1}.
\end{proposition}

\begin{proposition}[Backward-pass dominance in high dimensions]
\label{prop:backward_dominance}
For fixed output dimension $d_{\mathrm{out}}$ and large input dimension $d_{\mathrm{in}}$, the backward-pass shot cost scales as $\Theta(d_{\mathrm{in}})$ relative to the forward-pass cost.
\end{proposition}

\begin{corollary}[Backward-focused scaling analysis]
\label{cor:backward_focus}
In the regime $d_{\mathrm{in}} \gg 1$ with $d_{\mathrm{out}} = O(1)$ (typical for classification with few classes), dimension-dependent shot scaling is governed by backward-pass Jacobian estimation. Forward-pass uncertainty is retained for completeness but contributes only a lower-order term that does not affect the dimension-dependent scaling laws derived in Section~\ref{sec:singlestep}.
\end{corollary}
Proposition~\ref{prop:forward_variance}, Proposition~\ref{prop:backward_dominance}, and Corollary~\ref{cor:backward_focus} are proved in \ref{app:supp_derivations}.

\paragraph{Compatible observables.}
When quantum layer outputs consist of mutually commuting observables (e.g.\ $\{Z_i\}$ measured simultaneously), both forward and backward circuit evaluations reduce by a factor of $d_{\mathrm{out}}$ ($=C$ for the output layer), since all observables can be measured in a single circuit run. This optimization preserves the $\Theta(d_{\mathrm{in}})$ scaling ratio in Proposition~\ref{prop:backward_dominance}, as the reduction applies equally to numerator and denominator. In the commuting case the per-evaluation count drops from $2C$ to $2$, so the per-dimension budget becomes $s = 2\,s_0$ rather than $2C\,s_0$; the uniform formula $R = ds$ applies in either case (Section~\ref{sec:notation}).

\begin{remark}[Classical gradient cost and comparison]
\label{rem:classical_cost}
For a classical neural network with input dimension~$d$ and $P$ parameters, a single backpropagation pass computes the exact input gradient $\nabla_{x} L$ in $\Theta(P)$ floating-point operations. The forward pass also requires $\Theta(P)$ operations, giving a \emph{gradient cost ratio}
\[
  \rho_{\text{classical}}
  \;=\; \frac{C_{\text{bwd}}}{C_{\text{fwd}}}
  \;=\; O(1),
\]
independent of~$d$. We call $\rho$ the \emph{gradient cost ratio}: the cost of obtaining the gradient measured in units of one forward inference, that is, how many forward inferences a single gradient costs. This is a theorem, not merely an empirical observation: the Baur--Strassen result~\cite{baur_complexity_1983} and its modern treatment~\cite{griewank_evaluating_2008} (which terms the quantity the \emph{cost ratio}) prove that the cost of computing the full gradient vector via reverse-mode automatic differentiation is at most a small constant multiple ($\le 5{\times}$) of the forward pass for any differentiable program: the cheap-gradient principle, which bounds $\rho\le5$ for any classical program, independent of~$d$.

In contrast, for a quantum classifier the forward pass requires $R_{\text{fwd}} = \Theta(d_{\text{out}}/\sigma^2_{\text{fwd}})$ shots, independent of input dimension~$d$ (Proposition~\ref{prop:backward_dominance}), while a successful single-step gradient attack requires $R_{\text{bwd}} = \Theta(d^{2.5})$ shots (Proposition~\ref{prop:fgsm-optimization}, with $\epsilon \propto \sqrt{d}$ and $\|\gstar\|=\Theta(1)$). The quantum gradient cost ratio is therefore
\[
  \rho_{\text{quantum}}(d)
  \;=\; \frac{R_{\text{bwd}}}{R_{\text{fwd}}}
  \;=\; \Theta(d^{2.5}).
\]
The structural asymmetry is that classical forward and backward passes both scale with~$P$ (and hence with~$d$), keeping their ratio constant, whereas the quantum forward pass cost is $d$-independent (it depends only on the number of output classes $C$ and the desired precision) while the backward pass scales polynomially with~$d$. This decoupling makes the gradient cost ratio diverge: the cheap-gradient principle fails on quantum hardware, $\rho_{\text{quantum}}/\rho_{\text{classical}} = \Theta(d^{2.5})$ at the geometric floor, realized as $\Theta(d^{3.00})$ by the tested models (Section~\ref{subsec:classical_baseline}).
\end{remark}

\section{Shot-budget scaling laws}
\label{sec:scalinglaws}

\paragraph{Scope and relaxation of the working assumptions.}
Assumption~\ref{ass:1s_1} is valid when $\epsilon$ is small relative to the local radius of curvature $1/\|\nabla^2 L\|$. Under the natural scaling $\epsilon\propto\sqrt{d}$ motivated by norm concentration (\ref{app:gauss_l2}), when $d$ increases, the $\epsilon$-ball with increasing radius sweeps an increasingly larger portion of input space in absolute $\ell_2$ terms, and empirical evidence from classical adversarial ML suggests that loss-landscape curvature typically increases in high dimensions due to the proliferation of local modes and saddle points~\cite{dauphin2014saddle,li2018losslandscape}; second-order terms $O(\epsilon^2\|\nabla^2 L\|)$ therefore become non-negligible. The net empirical cost still exceeds this baseline, but direct measurement (\S\ref{par:gnorm_baseline}) locates the driver elsewhere: the dimensional decay of the gradient norm $\|\gstar\|\propto d^{-0.73}$, which the cost coefficient absorbs, accounts for the empirical excess over the baseline on both datasets (Section~\ref{sec:Res}). Local curvature itself, measured by a transverse line-scan of the $\epsilon$-step, mildly favors the attacker: the curved landscape forgives noise-induced misalignment, so the shot-noise penalty is smaller than the linear prediction ($\zeta<1$, \S\ref{par:gnorm_baseline}). The scaling laws below therefore remain shot-noise baselines whose empirical excess is gradient-norm-driven rather than curvature-driven. A relaxation that drops Assumptions~\ref{ass:1s_2}--\ref{ass:1s_3} and replaces them by a covariance condition $\tr\Sigma=\Theta(d)$ is developed in Section~\ref{subsec:expected_relax} and yields the same qualitative scaling.

\subsection{Single-step attacks}
\label{sec:singlestep}

We now derive the explicit scaling laws from the analytical framework of Section~\ref{sec:methods} (PSR unbiasedness and covariance, backward-pass cost dominance, and the single-step attack of Eq.~\eqref{eq:1s_attack} under Assumptions~\ref{ass:1s_1}--\ref{ass:1s_3}). Single-step methods such as FGSM draw a single noisy gradient and spend their entire perturbation budget along that direction, making them sensitive to gradient estimation errors: any misalignment induced by shot noise directly degrades attack efficacy. We derive explicit shot-scaling laws for both i.i.d.\ and correlated gradient-noise models. The key results are Propositions~\ref{prop:fgsm-optimization} and~\ref{prop:fgsm-success}, which show that maintaining constant attack efficacy requires $s=\Theta(\epsilon d)$ shots per input dimension, corresponding to a total gradient-evaluation budget $R=ds=\Theta(\epsilon d^{2})$, i.e.\ $\Theta(d^{2})$ for fixed $\epsilon$ and $\Theta(d^{5/2})$ under $\epsilon\propto\sqrt d$, 
and Corollary~\ref{cor:relaxed_shots}, which recovers the same scaling under weaker assumptions.

\subsubsection{Result: shot-budget guarantees for single-step attacks}

\paragraph{New shot-scaling law.}
\begin{proposition}[FGSM shot budget versus optimization gap]\label{prop:fgsm-optimization}
Under Assumptions~\ref{ass:1s_1}, \ref{ass:1s_2}, and~\ref{ass:1s_3}, the $l_2$ FGSM attack requires shots per dimension
\[
    s = O\!\left(\frac{\epsilon d}{\Delta\,\|\gstar\|}\right)
\]
and therefore total shots
\[
    d \; s = O\!\left(\frac{\epsilon d^2}{\Delta\,\|\gstar\|}\right)
\]
to achieve an optimization error $L(\delta) - L(\delta^*)=\Delta$. When $\|\gstar\|$ is locally constant, this simplifies to $s=\Theta(\epsilon d/\Delta)$ and $d \ s = \Theta(\epsilon d^2/\Delta)$.
\end{proposition}

\begin{proof}
We analyze the $\ell_2$-constrained FGSM attack, where the perturbation norm directly reflects the Euclidean structure of shot-noise estimation error (Sec.~\ref{sec:notation} and Sec.~\ref{sec:uncertain_gradient}). Under $\ell_2$ constraints, the optimal first-order perturbation aligns exactly with the gradient direction, enabling a clean analysis of how shot noise degrades attack effectiveness through angular misalignment. Consider the perturbation $\delta = -\epsilon \hat{g}/\|\hat{g}\|$ with measurement-driven gradient estimate $\hat{g}\sim N(\nabla L(x,y),I\sigma^2/s)$. Combining this with Assumption~\ref{ass:1s_1} yields
\begin{equation}
\label{eq:L_delta}
    L(x+\delta,y) = L(x, y) - \epsilon \|\nabla L(x,y)\| \cos(\alpha)\;,
\end{equation}
where $\alpha$ is the angle between $\hat{g}$ and the true gradient $\nabla L(x,y)$. Under Assumptions~\ref{ass:1s_1} and~\ref{ass:1s_2}, the normalized estimator $\hat{g}_{n} = \hat{g}/\|\hat{g}\|$ follows a von Mises--Fisher law with concentration parameter $\kappa_{\mathrm{vMF}}\approx s\|\gstar\|^2/\sigma^2$. Using the large-$\kappa_{\mathrm{vMF}}$ approximation from Assumption~\ref{ass:1s_3},
\begin{equation}
\label{eq:CosSim}
    \mathbb{E}[\cos \alpha] \approx 1-\frac{d-1}{2\kappa_{\mathrm{vMF}}} \approx 1-\frac{(d-1)\sigma^2}{2s\,\|\gstar\|^2}\;.
\end{equation}
Substituting into (\ref{eq:L_delta}) shows that, when the per-shot variance $\sigma^2$ is dimension-independent (the phase-dense encoding's generators have spectral gap $\Delta G_j=1$, making $\sigma^2$ flat by construction), the single-step loss surplus satisfies
\begin{equation}
    L(x+\delta,y) - L(x+\epsilon \gstar/\|\gstar\|,y) \approx \epsilon\,\frac{(d-1)\sigma^2}{2s\,\|\gstar\|}\;,
\end{equation}
and hence, using the local linearity benchmark, the optimization gap scales as
\begin{equation}
    L(\delta) - L(\delta^*) = O\!\left(\frac{\epsilon d}{s\,\|\gstar\|}\right)\;.
\end{equation}
Inverting the last expression to solve for the shots per dimension that achieve a target gap $\Delta$ yields the stated bounds.
\end{proof}

Proposition~\ref{prop:fgsm-optimization} yields a shot-noise scaling baseline $s=\Theta(\epsilon d)$, which evaluates to $\Theta(d^{3/2})$ under $\epsilon \propto \sqrt{d}$ and $\|\gstar\| = \Theta(1)$. The three assumptions underlying this result, local linearity (Assumption~\ref{ass:1s_1}), i.i.d.\ Gaussian noise (Assumption~\ref{ass:1s_2}), and small noise (Assumption~\ref{ass:1s_3}), are each optimistic for the attacker: local linearity means the gradient is maximally informative about the loss landscape, i.i.d.\ noise is the most favorable covariance structure, and the small-noise regime ensures gradient estimates are already high-quality. The $\|\gstar\|=\Theta(1)$ premise is a design property rather than an idealization: architectures built to avoid barren plateaus (quantum convolutional networks, whose gradients provably vanish at most polynomially in qubit count~\cite{pesah_absence_2021}; shallow circuits with local cost functions~\cite{Cerezo2021}; see~\cite{larocca_barren_2025} for a review) keep gradient magnitudes dimension-stable up to polylogarithmic factors in $d$, and for such models the $\Theta(d^{3/2})$ per-dimension baseline is the realized cost. Generic deep circuits without plateau mitigation enjoy an additional layer of protection: their gradient norms decay with dimension (measured directly in Section~\ref{sec:Res}: $\|\gstar\|\propto d^{-0.73}$, exponential in qubit count at $d\approx2^{q}$, the plateau signature), which the coefficient's $1/\|\gstar\|$ factor converts into the higher realized scaling $s \propto d^{2.00}$ on MNIST ($d^{2.07}$ on Fashion-MNIST), as detailed below.

\paragraph{Empirical scaling beyond the baseline.}
\label{par:empirical_surplus}
Proposition~\ref{prop:fgsm-optimization} predicts $s=\Theta(\epsilon d)=\Theta(d^{3/2})$ when its coefficient is evaluated under $\|\gstar\|=\Theta(1)$ (with $\epsilon\propto\sqrt d$). The sharpest experimental test of the law targets the coefficient with the gradient norm folded out: the measured combination $\mathrm{d}L\cdot\|\gstar\|$ follows $d^{1.46}$ on MNIST and $d^{1.42}$ on Fashion-MNIST against the parameter-free $d^{3/2}$ prediction (\S\ref{par:gnorm_baseline}), confirming the geometric factor of the bound. The bare shot requirement is then this geometric factor compounded by the model's gradient scale: direct measurement shows $\|\gstar\|\propto d^{-0.73}$ on MNIST ($d^{-0.65}$ on Fashion-MNIST) for our unmitigated deep circuits, and the resulting in-regime fits, $s \propto d^{2.00}$ and $s \propto d^{2.07}$ (Section~\ref{sec:Res}), require no parameterization beyond the bound's own $1/\|\gstar\|$ factor; landscape curvature contributes a small, oppositely-signed effect (a mild attacker advantage).

\begin{remark}[$\epsilon$-scaling]
If $\epsilon$ is dimension‑independent, maintaining constant FGSM efficacy requires $\Theta(d^{2})$ total shots.
If $\epsilon$ scales with dimension according to norm concentration (\ref{app:gauss_l2}), the requirement increases to $\Theta(d^{2.5})$ or higher.
\end{remark}

\paragraph{Shot-noise scaling baseline.}
\label{par:attacker_baseline}
Proposition~\ref{prop:fgsm-optimization} derives $R=\Theta(\epsilon d^{2})$ as the baseline scaling under Assumptions~\ref{ass:1s_1}--\ref{ass:1s_3}, which evaluates to $\Theta(d^{5/2})$ under $\epsilon\propto\sqrt d$, 
represents an idealized lower‑envelope scaling under a combination of assumptions that favor the attacker, and should therefore be interpreted as a baseline rather than a typical regime.
However, each assumption represents conditions that are favorable to the attacker: (i)~local linearity means the gradient is maximally informative about the loss landscape within the $\epsilon$-ball; (ii)~i.i.d.\ noise is the most favorable covariance structure for gradient alignment; (iii)~the small-noise regime ensures gradient estimates are already high-quality. For generic QML loss landscapes, where curvature grows with dimension, noise correlations are non-trivial, and gradient signal-to-noise varies across components, departing from these assumptions will typically increase the attacker's required shot budget. Corollary~\ref{cor:relaxed_shots} illustrates one such departure: under correlated noise with $\tr\Sigma = \Theta(d)$, the minimum sufficient total budget grows to $R_\mathrm{min}=\Theta(\epsilon^{2} d^{2})=\Theta(d^{3})$ under $\epsilon\propto\sqrt d$. The empirical excess over the baseline, $d^{0.50}$ (MNIST) and $d^{0.57}$ (Fashion-MNIST) per dimension (\S\ref{par:empirical_surplus}), is dominated by the measured gradient-norm decay (\S\ref{par:gnorm_baseline}): of the three idealizations, $\|\gstar\|=\Theta(1)$ is the load-bearing one, while local curvature, measured directly, mildly favors the attacker rather than the defender. Unlike the other two, $\|\gstar\|=\Theta(1)$ is attainable by construction: plateau-mitigated architectures realize it up to logarithmic factors~\cite{pesah_absence_2021,larocca_barren_2025}, whereas generic deep circuits violate it in the defender's favor. We therefore interpret $R=\Theta(\epsilon d^{2})$ (which gives $d^{5/2}$ under $\epsilon\propto\sqrt d$) as the \emph{shot-noise scaling baseline}: it quantifies what a single-step PSR-based attack costs against the most attacker-favorable model class, with departures driving costs higher; the plateau-prone circuits tested in Section~\ref{sec:Res} realize total exponents of $3.00$ (MNIST) and $3.07$ (Fashion-MNIST).

The result strengthens the robustness claim: even under the most favorable conditions for the attacker, the cost is already $R=\Theta(\epsilon d^{2})=\Theta(d^{5/2})$. In practice the measured gradient-norm decay drives it higher still (\S\ref{par:gnorm_baseline}), while local curvature returns a small offsetting advantage to the attacker. Whether alternative attack strategies can exploit that curvature further while respecting the $s=\Theta(\epsilon d)=\Theta(d^{3/2})$ shot-noise baseline remains an open question for future investigation.

\paragraph{Estimator-independence of the exponent.} The derivation uses only that each component is recovered by an unbiased estimator with per-component variance $\Theta(\sigma^2/s)$, not the two-shift structure of the parameter-shift rule; the $\Theta(d^2)$ exponent is therefore a property of unbiased gradient extraction, not of a particular estimator. Two routes could lower it, and both are closed in the expressive regime. Co-measuring several components at once is bounded by the expressivity--measurement-efficiency trade-off~\cite{chinzei_tradeoff_2025}, under which a near-maximal dynamical Lie algebra forces the number of simultaneously measurable components to $O(1)$. Estimating a directional derivative along a generic dense direction requires a quadrature exact on the circuit's full Fourier support, which for an expressive circuit is far larger than $d$, so it offers no saving over the $d$ axis-aligned probes. Parameter-shift is the canonical estimator that attains the bound~\cite{Wierichs2022}; biased surrogates are treated separately in \ref{app:zero-order}.

The preceding i.i.d.\ analysis and the success-probability analysis that follows (Section~\ref{subsec:success_prob}) provide two complementary paths to the same quadratic scaling conclusion: the first relaxes distributional assumptions, the second strengthens the guarantee from expectation to tail probability.

\subsubsection{Result: correlated-noise guarantee}
\label{subsec:expected_relax}
\begin{corollary}[Minimum sufficient shot budget under correlated noise]
\label{cor:relaxed_shots}
Fix a target loss-surplus tolerance $\Delta$ (written $\Delta$ throughout, and not to be confused with the compound symbols $\Delta L$, $\Delta G_i$) on the expected single-step loss surplus. Consider an attacker who allocates a per-dimension budget of $s$ shots and seeks to drive the expected single-step loss surplus to $\mathbb{E}[\Delta L]\le\Delta$. Because the surplus bound is monotonically decreasing in $s$ (proof below), this target is met for every allocation at or above the \emph{minimum sufficient budget}
\begin{equation}
    s_\mathrm{min}(\Delta) \;=\; \frac{4\,\epsilon^2\,\tr\Sigma}{\Delta^2}\,;
\end{equation}
that is, any $s\ge s_\mathrm{min}$ certifies the target tolerance $\Delta$. More generally, an arbitrary per-dimension allocation $s$ certifies the tolerance
\begin{equation}
    \Delta_\mathrm{cert}(s) \;=\; 2\,\epsilon\,\sqrt{\frac{\tr\Sigma}{s}}\,,
\end{equation}
which meets the target, $\Delta_\mathrm{cert}(s)\le\Delta$, precisely when $s\ge s_\mathrm{min}$. If $\tr\Sigma$ scales linearly with dimension, i.e., $\tr\Sigma=\Theta(d)$, then a dimension-explicit budget that certifies $\Delta$ is
\begin{equation}
    s_\mathrm{min} \;=\; \Theta\!\left(\frac{\epsilon^2}{\Delta^2}\, d\right),\qquad
    R_\mathrm{min}=d\times s_\mathrm{min} \;=\; \Theta\!\left(\frac{\epsilon^2}{\Delta^2}\, d^2\right).
\end{equation}
Thus, under this relaxed assumption, the \emph{total} gradient budget required to hold expected FGSM efficacy constant grows quadratically in $d$.
\end{corollary}

The correlation-agnostic alignment bound and the minimum-budget computation are given in \ref{app:supp_derivations}.

\paragraph{Discussion.} Compared to the i.i.d.-based vMF analysis (which yields a $\Theta(d/s)$ expected cosine loss $1-\mathbb{E}[\cos\alpha]$), the relaxed expected bound scales as $\Theta(\sqrt{d/s})$ when $\tr\Sigma=\Theta(d)$. Thus, keeping $\mathbb{E}[\Delta L]$ fixed forces a per-dimension \emph{minimum budget} that scales as $s_\mathrm{min}=\Theta(\epsilon^2 d)$ and a total $R_\mathrm{min}=\Theta(\epsilon^2 d^2)$, a conservative, assumption-light provisioning rule. In our experiments we set $\epsilon\propto d^{0.5}$; the fitted in-regime shot law is $s\propto d^{2.00}$, which equivalently can be written as $s\propto \epsilon\,d^{1.50}$ after factoring out the imposed $d^{0.5}$ scaling of $\epsilon$. This empirical exponent coincides with the relaxed minimum-budget law $s\propto \epsilon^2 d$ (i.e., $d^2$ under the same $\epsilon$) and lies below the i.i.d./vMF law with the measured gradient norm inserted ($s\propto \epsilon d/\|\gstar\|$, i.e., $d^{2.23}$ given $\|\gstar\|\propto d^{-0.73}$), the shortfall reflecting the measured landscape factor $\zeta(d)<1$ and the extraction sensitivity of the fit. Table~\ref{tab:scaling} summarises these regimes side by side.

\begin{table}[t]
\centering
\caption{\textbf{FGSM shot-scaling regimes.} Each row states the setting, the resulting per-dimension and total shot scaling, the assumptions required, and the source result.}
\label{tab:scaling}
\scriptsize
\begin{tabular}{lllll}
\toprule
\textbf{Setting} & $s(d)$ & $R=\Theta(d{\cdot}s)$ & \textbf{Assumptions} & \textbf{Source}\\
\midrule
Fixed $\epsilon$, $\|\gstar\|{=}\Theta(1)$, const.\ $\tau$
  & $\Theta(d)$ & $\Theta(d^2)$
  & Asm.~\ref{ass:1s_1}--\ref{ass:1s_3}
  & Prop.~\ref{prop:fgsm-success}\\
$\epsilon{\propto}\sqrt{d}$, $\|\gstar\|{=}\Theta(1)$
  & $\Theta(d^{3/2})$ & $\Theta(d^{5/2})$
  & Asm.~\ref{ass:1s_1}--\ref{ass:1s_3}
  & Prop.~\ref{prop:fgsm-optimization}\\
$\epsilon{\propto}\sqrt{d}$, $\|\gstar\|{=}\Theta(\sqrt{d})$
  & $\Theta(d)$ & $\Theta(d^2)$
  & Asm.~\ref{ass:1s_1}--\ref{ass:1s_3}
  & Prop.~\ref{prop:fgsm-optimization}\\
Correlated noise, $\tr\Sigma{=}\Theta(d)$, $\epsilon{\propto}\sqrt{d}$
  & $O(d^{2})$ & $O(d^{3})$
  & $\tr\Sigma{=}\Theta(d)$ only
  & Cor.~\ref{cor:relaxed_shots}\\
Empirical MNIST ($\|\gstar\|{\propto}d^{-0.73}$ measured)
  & ${\approx}\,d^{2.00}$ & ${\approx}\,d^{3.00}$
  & Measured; in-regime fit
  & \S\ref{par:empirical_surplus}, Sec.~\ref{sec:Res}\\
\bottomrule
\end{tabular}
\end{table}

\subsubsection{Result: success probabilities under measurement uncertainty}
\label{subsec:success_prob}
Having established shot-scaling requirements for expected loss in the preceding sections, we now derive explicit success-probability bounds. While expected-loss analysis captures average attack degradation, security guarantees often require bounding the probability that an attack succeeds on any given sample. Here, we translate the concentration properties of noisy gradient estimates into probabilistic robustness guarantees.

\begin{figure}[!htbp]
    \centering
\begin{tikzpicture}[>=Stealth, scale=1.6]

  \coordinate (X)    at (1.000, 0.800);
  \coordinate (Foot) at (2.096, 1.896);
  \coordinate (B1)   at (3.864, 0.128);
  \coordinate (B2)   at (0.328, 3.664);
  \coordinate (Gtip) at (3.121, 2.921);
  \coordinate (Ghat) at (3.762, 1.972);
  \coordinate (Xadv) at (2.841, 1.581);

  \fill[purple!8] (B2) -- (B1) -- (3.864, 3.664) -- cycle;

  \draw[purple, very thick] (B1) -- (B2)
    node[pos=0.88, above, sloped, purple, font=\small] {decision boundary};

  % perpendicular distance l_x (collinear with g*)
  \draw[thin, gray] (X) -- (Foot);
  \draw[thin, gray] (1.926, 1.896) -- (2.011, 1.811) -- (2.096, 1.726);
  \node[font=\small, gray, left=1pt] at (1.45, 1.35) {$l_x$};

  % tolerance cone of half-angle tau around g*: edge lines + a vertex arc make tau read as an angle
  \draw[densely dotted, gray, thick] (X) -- ([shift={(84.2:2.0)}] X);
  \draw[densely dotted, gray, thick] (X) -- ([shift={(5.8:2.0)}] X);
  \draw[densely dotted, gray, thick]
    ([shift={(5.8:2.0)}] X) arc (5.8:84.2:2.0);
  % tau measured as the half-angle from g* (45 deg) to the upper cone edge (84.2 deg)
  \draw[gray, thick] ([shift={(45:1.2)}] X) arc (45:84.2:1.2);
  \node[font=\small, gray] at ([shift={(64.6:1.46)}] X) {$\tau$};

  \draw[->, thick, blue] (X) -- (Gtip)
    node[pos=1.0, above left=-1pt, font=\small, blue] {$\gstar$};

  \draw[->, thick, red!70!black, dashed] (X) -- (Ghat)
    node[pos=1.0, right=1pt, font=\small, red!70!black] {$\hat{g}$};

  \draw[->, very thick, orange] (X) -- (Xadv)
    node[pos=0.5, below=1pt, font=\small, orange] {$\epsilon$};
  \fill[orange] (Xadv) circle (1.8pt)
    node[below right=-2pt, font=\small] {$x{+}\delta$};

  \draw[thick] ([shift={(23.0:0.7)}] X) arc (23.0:45.0:0.7);
  \node[font=\small] at ([shift={(34.0:0.95)}] X) {$\alpha$};

  \fill (X) circle (2pt) node[below left, font=\small] {$x$};
  \node[font=\footnotesize, purple!70!black, anchor=north east] at (3.75, 3.55) {misclassified};
  \node[font=\footnotesize, anchor=west] at (0.05, 2.0) {correct class};
\end{tikzpicture}
    \caption{\textbf{Geometry of a single-step shot-noise attack.} The sample~$x$ lies at perpendicular distance $l_x$ from a locally linear decision boundary, with the true gradient $\gstar$ (blue) pointing toward it. Shot noise deflects the estimated gradient $\hat{g}$ (red dashed) by angle $\alpha$; the perturbation of strength $\epsilon$ (orange) follows $\hat{g}$. The attack succeeds (i.e., $x{+}\delta$ crosses the boundary) only when $\alpha \le \tau = \arccos(l_x/\epsilon)$. The dotted lines bound the tolerance cone of half-angle $\tau$ about~$\gstar$ (its edges reach the boundary at radius $\epsilon$); the grey arc marks $\tau$ measured from~$\gstar$, so the attack succeeds when $\hat{g}$ falls inside the cone.}
    \label{fig:angle2robustness}
\end{figure}
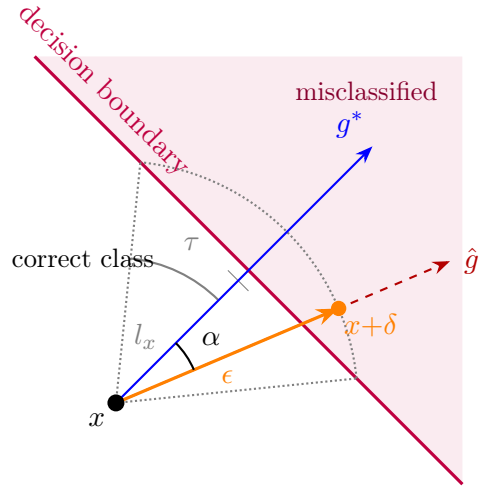

The geometry of Fig.~\ref{fig:angle2robustness} converts the angular concentration of the noisy gradient into a success probability. Conditioning on an attackable sample at distance $l_x\le\epsilon$ from a locally linear boundary, the single-step attack $\delta=-\epsilon\hat{g}$ succeeds exactly when the estimate falls within the tolerance cone $\alpha\le\tau=\arccos(l_x/\epsilon)$. Because the normalized estimate is von Mises--Fisher with concentration $\kappa_{\mathrm{vMF}}\approx s\|\gstar\|^2/\sigma^2$, the success probability is an incomplete-Gamma tail that concentrates, for large $d$, around $\sqrt{s}\,\tau\|\gstar\|-\sigma\sqrt{d}=\text{const}$. Holding the success rate fixed (e.g.\ at $50\%$) therefore forces $s=O(d)$ per dimension and $R=O(d^2)$ in total: the same quadratic law as the expectation analysis of Section~\ref{subsec:expected_relax}, now as a tail guarantee rather than an expected-loss bound (Proposition~\ref{prop:fgsm-success}; full derivation in \ref{app:supp_derivations}).

\subsection{Multi-step attacks: resource scaling for iterative optimization}
\label{sec:multistep}
Iterative attacks amortize the per-step FGSM cost only partially: each iteration refines the perturbation but requires a fresh gradient estimate, so the shot budget grows linearly with the number of steps $K$. Under a fixed total budget $R=Kds$ the attacker therefore faces a bias--variance--iteration trade-off in which lowering $s$ to buy more steps degrades per-step gradient quality.

We study the Carlini \& Wagner (C\&W) attack~\cite{carlini_towards_2016} as a representative strong iterative attack, because its weighted-sum objective $c\,L_{mc}(x+\delta)+\|\delta\|_2^2$ maps cleanly onto Stochastic Gradient Langevin Dynamics (SGLD), for which non-asymptotic convergence theory exists~\cite{raginsky_non-convex_2017, xu_global_2020}. Writing the finite-shot C\&W update as SGLD identifies the per-iteration shot count with an effective inverse temperature, $\beta = 2s/(\eta\sigma^2)\propto s$ (\ref{app:sgld_sufficient}): a low shot budget gives a high-temperature, diffusive trajectory, a high budget a low-temperature, greedy descent. This $\beta\propto s$ mapping is the modeling framework we use to read the iterative experiments; it predicts the bias--variance trade-off and the empirical shot sweet spot $s^*(d)$ reported in Section~\ref{subsec:multistep}.

Combining the mapping with Raginsky et al.'s convergence theorem yields a \emph{sufficient} total-shot budget, polynomial in $d$, for convergence below a target loss gap. The operative consequence is that finite-shot convergence theory gives an iterative attacker no guarantee of escaping the single-step scaling within any feasible budget. The explicit degree of this budget, the assumptions and their verification for a canonical PQC/softmax/C\&W model, and the sense in which it is a conservative sufficient condition are derived in \ref{app:sgld_sufficient}.

\section{Numerical and hardware experiments}
\label{sec:experiments}

\subsection{Simulation across architectures and datasets \label{sec:Res}}
Theory predicts a $d^2$ floor for single-step attacks and an empirical scaling excess that exceeds it; the SGLD framework further suggests a bias--variance trade-off for iterative methods. We now test each prediction. Each experimental theme is paired with the theoretical takeaway it interrogates: (i) whether the $\Theta(\epsilon d/s)$ single-step law from Proposition~\ref{prop:fgsm-optimization} appears in practice as model and data complexity grow; (ii) whether the same scaling persists under changes of the circuit ansatz and the gradient estimator; and (iii) how the SGLD modeling framework of \ref{subsec:sgld_summary} manifests as an observable shot-budget sweet spot.

Datasets and phase-dense encoding, model architecture and training, attack normalization ($\epsilon=\kappa\sqrt{d}$, PSR with $s$ shots per input dimension and total budget $R=ds$), the iterative-attack convention ($K=5$ steps with $s=R/(Kd)$), and the per-$(q,s)$ evaluation protocol are reported in \ref{methods:sim_common}; the design philosophy that motivates this single-factor-at-a-time setup is summarized in \ref{methods:sim_common} (\emph{Scope and generality}).
Taken together, these protocol choices (PSR with budget $R=ds$, the $\epsilon=\kappa\sqrt{d}$ schedule, fixed iteration count $K$, and a single experimental factor varied at a time) ensure that the estimated empirical gradients follow the unbiased, shot-limited model used throughout the theory section: PSR supplies the unbiased estimator with variance $\Theta(1/s)$ assumed in Assumption~\ref{ass:1s_2}, the $\epsilon=\kappa\sqrt{d}$ schedule instantiates the perturbation scaling invoked in the proofs of Proposition~\ref{prop:fgsm-optimization}, and keeping $J_L$ and $J_U$ fixed isolates the role of the quantum-layer covariance in~(\ref{eq:variance1}). Consequently, each experiment can be interpreted as a concrete instantiation of the theoretical resource laws.

Each subsequent subsection varies exactly one experimental factor while holding all others fixed: Section~\ref{subsec:shot_scaling} varies $(q,d)$ and dataset choice to study how the single-step scaling laws manifest empirically. Two robustness checks, swapping the circuit ansatz (\ref{app:arch_general}) and replacing PSR with finite-difference estimators (\ref{app:estimator}), confirm that the observed scaling is intrinsic to measurement-based gradient estimation, holding across architectures and estimators; detailed plots and discussion are deferred to the appendices. Finally, Section~\ref{subsec:multistep} turns to multi-step attacks (C\&W), examining how the per-iteration shot count $s$ sets the effective SGLD temperature and thereby controls the gradient-noise--fidelity trade-off at fixed $K$; the experiments reveal a dimension-dependent sweet spot $s^*(d)$ that drifts upward with $d$, consistent with the noise-temperature picture of the SGLD framework of \ref{subsec:sgld_summary}.

\subsubsection{Shot scaling with dimension and dataset complexity}\label{subsec:shot_scaling}
Guided by Proposition~\ref{prop:fgsm-optimization}, we test whether the predicted $1/s$ dependence between per-dimension shot budget and single-step loss surplus holds empirically, and how the coefficient $H(d)$ scales with dimension. The small-noise analysis behind the proposition fixes the structure of that coefficient: the loss gap obeys
\[
\Delta L(s)\;=\;L(\delta)-L(\delta^*)\;\approx\;\frac{H(d)}{s},\qquad
H(d)\;=\;\frac{\epsilon\,(d-1)\,\sigma^{2}\,\zeta(d)}{2\,\|\gstar\|(d)},
\]
where $\zeta(d)$ is the landscape factor, the ratio of the realized clean-attack loss rise to its first-order prediction, measured directly in \S\ref{par:gnorm_baseline} ($\zeta(d)=1$ for a locally linear loss). The $H(d)\propto\epsilon d\propto d^{3/2}$ baseline of Proposition~\ref{prop:fgsm-optimization} is this rule evaluated under $\|\gstar\|=\Theta(1)$ and $\zeta(d)=1$; the experiments below measure each factor separately. Using the shared configuration (Sec.~\ref{sec:Res}), we sweep $q\in\{4,5,6,7,8,9,10\}$ on MNIST and Fashion-MNIST under $\epsilon=\kappa\sqrt{d}$. Shot budgets are placed at log-spaced multiples $h=s/s^\star$ of the critical count $s^\star(d)=(d-1)\sigma^2/2\|\gstar\|^2$ at which the gradient signal-to-noise ratio reaches unity, reaching $h=12$ at every dimension including $d=784$, well inside the small-noise regime where the $1/s$ law and the CLT approximation behind Assumption~\ref{ass:1s_3} hold; the fit window starts at the critical count itself ($h=1$), where that ratio equals one. ($\|\gstar\|$ and $\sigma^2$ enter the grid construction through preliminary probe measurements that are not part of the reported curves; \ref{methods:sim_common}.)

Figure~\ref{fig:waterfall_mnist} shows the resulting accuracy and loss curves versus shot budget and model size for both datasets. Two trends emerge: (1)~within fixed $q$, the shortfall from the perfect-gradient attack decays as $1/s$, matching $O(\epsilon d/s)$; (2)~across increasing $q$, larger models require proportionally more shots per dimension for the same shortfall.

To extract $H(d)$, we fit the loss gap through the origin, $\Delta L(s) = B_q/s$, on the in-regime points $h\ge1$. Figure~\ref{fig:coeffic_vs_d} plots $B_q$ versus $d$ on log-log axes; power-law fits yield $H(d)\propto d^{2.00}$ (95\% CI $[1.91,2.08]$) for MNIST and $H(d)\propto d^{2.07}$ (95\% CI $[1.99,2.15]$) for Fashion-MNIST, both with $R^2=0.999$ (unweighted log--log OLS over the seven dimensions $d\in\{16,25,64,121,256,484,784\}$; CIs are Student-$t$ at $n{-}2{=}5$ degrees of freedom). The exponent is robust to the extraction method: tightening the window to $h\ge3$ or fitting the saturating form $A/(s+B)$ to all points moves it within $[1.96,2.00]$ on MNIST and $[2.04,2.07]$ on Fashion-MNIST. The bare coefficient sets the attacker's budget, but it is not the theory's sharpest observable: by the coefficient rule it compounds the geometric shot-noise factor with the model's $1/\|\gstar\|$, so its exponent is a joint property of the estimator and the model's gradient scale. The next paragraphs separate the two by direct measurement; the consistent trends across both datasets already indicate that the scaling is embedded in the gradient estimator and persists across datasets.

\begin{figure}[t]
    \centering
    \includegraphics[width=\textwidth]{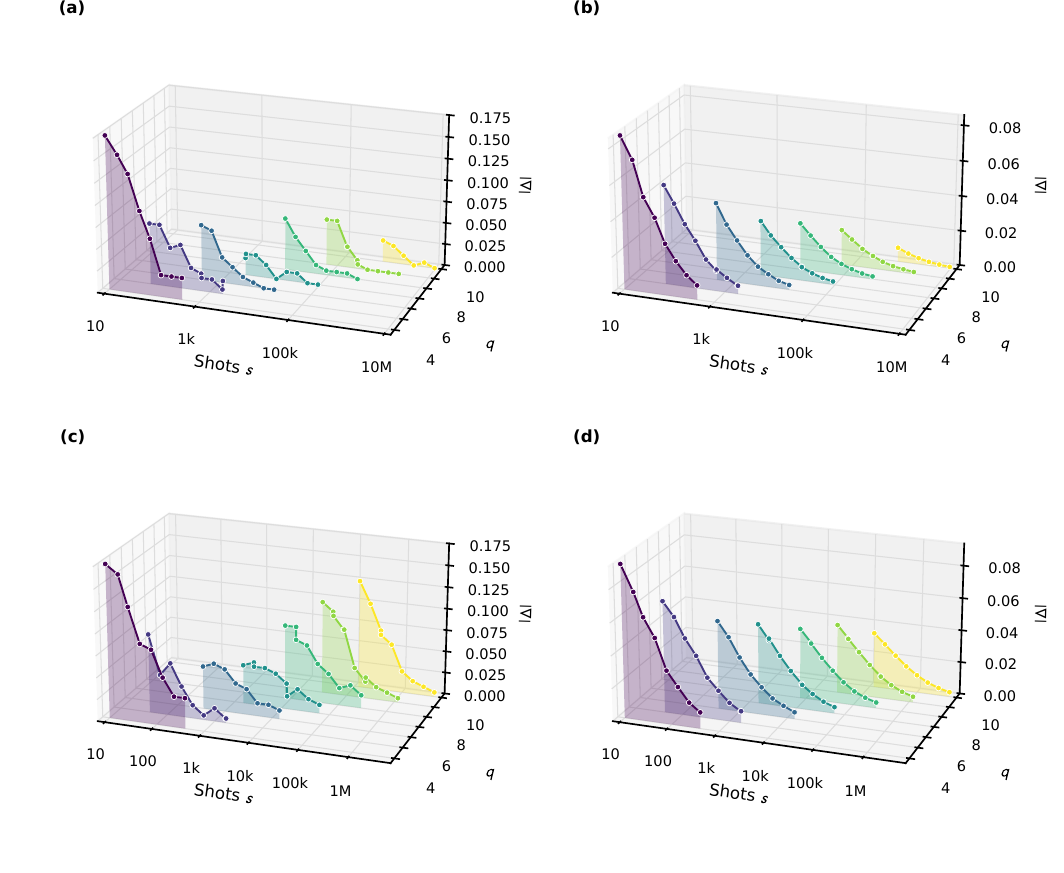}
    \caption{\textbf{FGSM shot-scaling on MNIST and Fashion-MNIST.} Waterfall plots versus shot budget and model size, on log-spaced grids at fixed multiples $h=s/s^\star$ of the critical count, reaching $h=12$ at every $q$. Axes: x = shots per input dimension $s$ (log scale, including both PSR shifts); depth = qubit/class count $q\in\{4,\dots,10\}$; height = absolute shortfall of the $s$-shot attack from the perfect-gradient attack; a marker is drawn at every measured grid point. MNIST is shown in (\textbf{a}) accuracy and (\textbf{b}) adversarial loss, and Fashion-MNIST in (\textbf{c}) accuracy and (\textbf{d}) adversarial loss. Within fixed $q$, the shortfall decays as $1/s$, matching the $O(\epsilon d/s)$ scaling; across increasing $q$, curves decay more slowly, requiring larger shots per dimension for constant attack efficacy. The two datasets show the same qualitative trends, confirming that the $1/s$ scaling law and the superlinear growth of $H(d)$ are properties of the quantum gradient estimator rather than of a particular dataset.}
    \label{fig:waterfall_mnist}
\end{figure}

\begin{figure}[t]
    \centering
    \includegraphics[width=0.66\linewidth]{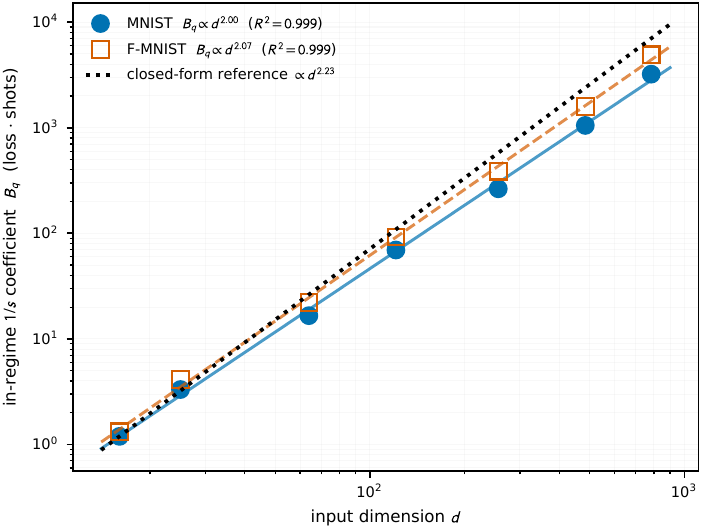}
    \caption{\textbf{Per-dimension shot-scaling coefficients.} Log--log plot of the through-origin coefficients $B_q$ (fit window $h\ge1$) versus input dimension $d$. MNIST: $B_q\approx 4.7\times10^{-3}\,d^{2.00}$ ($R^2=0.999$); Fashion-MNIST: $B_q\approx 4.5\times10^{-3}\,d^{2.07}$ ($R^2=0.999$). Dotted: the aggregate reference $\propto d^{2.23}$ obtained by inserting the cohort-mean gradient-norm decay into the coefficient rule (\S\ref{par:gnorm_baseline}). Both datasets exhibit superlinear growth far above the idealized $d^{1.5}$ baseline, with Fashion-MNIST slightly steeper.}
    \label{fig:coeffic_vs_d}
\end{figure}

\paragraph{Gradient-norm baseline.}\label{par:gnorm_baseline} The coefficient rule $H(d)=\epsilon(d-1)\sigma^2\zeta/2\|\gstar\|$ carries the gradient norm in its denominator, so the bare exponent of $B_q$ conflates the shot-noise geometry with the scaling of $\|\gstar\|$ itself. We measure that ingredient directly: the FGSM gradient norm decreases with dimension, $\|\gstar\|(d)\propto d^{-0.73}$ on MNIST and $\propto d^{-0.65}$ on Fashion-MNIST (opposite to the $\sqrt{d}$ growth of the classical intuition), while the per-shot readout variance stays flat at $\sigma^2\approx3$ (measured through $d=121$; assumed flat above). The decay is the expected behavior for this model class: at $d\approx2^{q}$ it is exponential in qubit count, the barren-plateau signature of deep generic ans\"atze~\cite{larocca_barren_2025}, and our test circuits employ no plateau mitigation. Folding the measured gradient norm back into the coefficient, per sample, forms the gradient-norm-independent combination
\begin{equation}\label{eq:dLg}
  \mathrm{d}L\cdot\|\gstar\| \;=\; \tfrac{1}{2}\,\epsilon\,(d-1)\,\sigma^2\,\zeta(d),
\end{equation}
in which $\|\gstar\|$ cancels exactly: what remains is the parameter-free geometric shot-noise cost $\epsilon(d-1)\sigma^2/2\propto d^{3/2}$, modulated only by the landscape factor $\zeta(d)$. Evaluating Eq.~\eqref{eq:dLg} per sample with the measured $\zeta_i$ gives $d^{1.46}$ on MNIST; the aggregate construction $B_q\cdot\langle\|\gstar\|\rangle$ on Fashion-MNIST gives $d^{1.42}$ (Fig.~\ref{fig:coeff_gnorm}). Both recover the $d^{3/2}$ baseline of Proposition~\ref{prop:fgsm-optimization}: the superlinear excess of the bare coefficient over $d^{3/2}$ is the measured $\|\gstar\|$ decay, not a property of the noise geometry. The per-sample construction matters because the ingredients are correlated across samples ($\mathrm{corr}(\|\gstar\|,\sigma^2)\approx+0.75$); the absolute scale of the semi-analytic per-sample gap agrees with the empirically fitted one to within $0.71$--$0.87$, the shortfall being the through-origin extraction reading low at finite $h$.

\paragraph{Takeaway: the single-step exponent.} The experiments validate the single-step law at the level of its structure. Every factor of the small-noise rule $H(d)=\epsilon(d-1)\sigma^2\zeta/2\|\gstar\|$ is measured independently, and the parameter-free geometric prediction is confirmed: with the gradient norm folded out, the coefficient follows $d^{1.46}$ (MNIST) and $d^{1.42}$ (Fashion-MNIST) against the predicted $d^{3/2}$, the small residual carried by the measured landscape factor $\zeta(d):1.05\to0.83$ (a mild, dimension-growing attacker advantage). The simulator results are thus accounted for with no free parameters beyond the measured model properties. The bare fitted coefficient is this geometry compounded by the gradient-norm decay: for the tested plateau-prone circuits, maintaining fixed single-step efficacy costs $s\propto d^{2.00}$ per dimension and $R=ds\propto d^{3.00}$ (95\% CI $[2.91,3.08]$) in total on MNIST ($d^{3.07}$, $[2.99,3.15]$, on Fashion-MNIST), while plateau-mitigated models with $\|\gstar\|=\Theta(1)$ realize the geometric floor $R=\Theta(d^{5/2})$. An aggregate estimate that inserts the cohort-mean gradient norm, instead of averaging per sample, brackets this $d^{2.00}$ from above at $d^{2.23}$ (\ref{app:supp_figs}). Three reference exponents recur through the paper: the gradient-norm-free shot-noise geometry $d^{3/2}$ (per dimension, measured $d^{1.46}$); the total-budget floor $R=\Theta(d^{5/2})$ reached by plateau-mitigated models; and the realized total budget $R\approx d^{3.00}$ of the tested gradient-decaying circuits, the floor inflated by their $\|\gstar\|$ decay. The total budget thus grows fast enough to erode the practicality of high-accuracy gradient attacks on large QML models.

\paragraph{Architectural robustness.}
To verify that these scaling laws are not artifacts of a particular circuit ansatz, we repeat the $(q,s)$ sweep under the same grid protocol for a data re-uploading variant that encodes the input twice, at the input and at mid-depth, within a single 200-layer circuit, for $q\in\{4,5,6,7\}$. The in-regime fit yields $c_0\propto d^{1.99}$, statistically indistinguishable from the single-upload $d^{1.96}$ over the matched dimension range. The coefficient itself is also unchanged (pre-factor ratio $0.98$): the second encoding pathway raises both the gradient norm ($\approx1.6\times$) and the per-shot readout variance ($\approx1.6\times$), and the two cancel in the coefficient rule $H\propto\sigma^2/\|\gstar\|$, a direct architectural confirmation of that structure. The residual architectural cost is the evaluation count: each input component appears in two encoding blocks, so a PSR gradient costs four shifted circuit evaluations per dimension instead of two, doubling the total budget per gradient at equal $s$. Architecture thus modulates the \emph{constants} of the cost via the $J_L$/$J_U$ factors and the encoding multiplicity rather than the fundamental dimension dependence. Full waterfall plots, coefficient fits, and detailed discussion are provided in \ref{app:arch_general}.

\paragraph{Estimator comparison.}
\ref{app:estimator} compares PSR against finite-difference (FD) and SPSA estimators under matched shot budgets. The key finding is that estimator choice affects constants and low-shot crossovers but not the fundamental scaling: Being an unbiased estimator, PSR dominates at moderate/large $s$. Full plots and discussion are given in \ref{app:estimator}.

\paragraph{Practical cost at scale.}
To see what the two-layer cost structure implies in practice, consider a single-step FGSM attack on a $d=784$ ($q{=}10$) MNIST classifier.
The smallest measured budget that brings the attack within ${\sim}15\%$ of its perfect-gradient accuracy drop is $s\approx7\times10^{5}$ shots per input dimension ($=2\times3.4\times10^{5}$, the two PSR shifts), above the critical count $s^\star(784)\approx4\times10^{5}$ set by the measured gradient norm.
Because all output observables are mutually commuting $Z_i$ measurements, they are obtained from the same circuit run (see the compatible-observables discussion in Section~\ref{methods:fp_bp_dominance}), so the total shot budget is $R = ds \approx 5\times 10^{8}$ per adversarial sample.
Assuming a typical circuit execution time of ${\sim}100\;\mu\mathrm{s}$ per shot, crafting a \emph{single} adversarial example requires ${\sim}5\times 10^{4}\;\mathrm{s} \approx 15\;\mathrm{hours}$ of sequential device time.
Mounting a dataset-wide attack on the full MNIST test set ($10{,}000$ images) would therefore consume ${\sim}1.5\times 10^{5}\;\mathrm{hours} \approx 17\;\mathrm{years}$ of continuous device time.
Parallelization can reduce wall-clock time but not the total number of shots required; the aggregate quantum resource cost remains the same.
These estimates further assume zero queuing, calibration, or communication overhead; real deployments would be slower still.

\subsubsection{Multi-step attacks and optimal shot allocation}\label{subsec:multistep}
For iterative C\&W attacks at fixed iteration count $K$, the per-iteration shot count $s$ is the sole control on gradient noise, setting the SGLD temperature $\beta\propto s$ through Eqs.~(\ref{eq:SGLD}) and~(\ref{eq:ChangeOfVar}). Low $s$ (high temperature) makes each step imprecise but lets the stochasticity dislodge the iterate from shallow minima of the C\&W landscape; high $s$ (low temperature) gives accurate steps that, within a fixed $K$, more readily settle into the nearest local trap. An optimal budget $s^*(d)$ therefore emerges from this fidelity-versus-exploration trade-off, which the SGLD framework (\ref{app:sgld_sufficient}) predicts qualitatively through $\beta$. When $K$ is instead free under a fixed total $R=Kds$, lowering $s$ also buys more iterations, the classical bias--variance trade-off; holding $K$ fixed isolates the noise-driven escape effect probed below.

We examine C\&W under the shared $K=5$ iteration budget and $\epsilon=\kappa\sqrt{d}$. The per-iteration shots $s$ are swept for qubit counts $q\in\{4,5,6,7\}$, with adversarial effectiveness recorded as in the single-step analysis (post-attack accuracy and loss surfaces). C\&W loss values are not compared across $q$ because the tuned coefficient $c$ is sample dependent. Tuning overhead is excluded from $R$, so any practical search would only strengthen the observed scaling.

Empirically, increasing $s$ from low values strengthens the attack: post-attack accuracy falls toward the exact-gradient (infinite-shot) value, and the budget $s^*(d)$ at which it saturates grows steeply with dimension, from $s^*\!\approx\!6\times10^{2}$ at $d{=}25$ to $\approx\!4\times10^{4}$ at $d{=}121$, nearly two orders of magnitude and broadly consistent with the $s\propto d^{2}$ growth of the single-step law. Figure~\ref{fig:waterfall_cw_asr} plots the signed accuracy gap to that exact-shot attack; because the reference shares the swept points' search budget, the gap is a clean measure of shot noise alone. In three of the four models ($d\in\{16,25,121\}$) the optimum drops \emph{below} the exact-shot value: a moderate amount of shot noise makes the attack stronger than an exact gradient, so a noisy gradient step is here preferable to the exact one. This is the stochastic-regularization effect the SGLD picture predicts, the shot noise playing the role of Langevin exploration that escapes shallow minima of the C\&W objective. The sub-floor depth is largest at small $d$ ($0.9$ pp at $d{=}25$, $0.3$ pp at $d{=}16$) and narrows to $0.13$ pp at $d{=}121$ (all resolved at $>4\sigma$ by averaging five shot-noise realizations); only $d{=}64$ shows no sub-floor dip, descending monotonically to the exact-shot floor. Simple warm‑up schedules that raise $s$ over iterations can outperform uniform allocation at the same $R$, echoing classical SGLD heuristics. The four tested dimensions are too few to fit a precise exponent for the iterative case.

Because $s^*$ climbs steeply with $d$ while the sub-floor benefit shrinks, the budget an iterative attacker must spend to reach the exact-gradient attack, let alone improve on it, grows rapidly with dimension, mirroring the single-step cost law. The trend holds across the dimensions tested at fixed $K$ and step size $\eta$.

\paragraph{Takeaway.} A dimension-dependent optimum $s^*(d)$ emerges from the gradient-noise--fidelity trade-off at fixed $K$: the SGLD temperature $\beta\propto s$ must be high enough for the chain to escape shallow local minima yet low enough for gradient updates to be constructive; we fix $K$ and $\eta$ instead of optimizing schedules and exclude hyperparameter-tuning overhead (e.g., binary search on $c$ in C\&W), so accounting for such costs would only make attacks more expensive in practice. Across the four dimensions tested, both $s^*(d)$ and $R_{\text{crit}}(d)$ rise with $d$ in a manner qualitatively consistent with the SGLD framework of \ref{sec:CW}. Simple warm‑up/annealed shot schedules can improve effectiveness at the same $R$. A precise empirical scaling exponent for iterative attacks would require a broader dimension sweep and is left to future work; the SGLD analysis (\ref{app:sgld_sufficient}) supplies the modeling framework that accounts for this trend.

\subsubsection{Classical baseline comparison}\label{subsec:classical_baseline}

The preceding quantum scaling laws quantify how attack cost grows with~$d$, but the force of this result depends on the contrast with the classical case. We now provide that anchor by training matched classical neural networks on the same dimension ladder, perturbation budget, and datasets, measuring both the gradient cost ratio and the attack success rate.

\paragraph{Classical model and timing protocol.}
The matched VGG-style CNN architecture (channel base $c=2q$, parameter counts comparable to the quantum circuits), the perturbation budget $\epsilon=\sqrt{d}/32$ matched to the quantum experiments, and the per-sample CPU wall-clock timing protocol (100 warmup iterations, 500~samples $\times$ 5~repetitions, median reported) used to compute $\rho_{\text{classical}}=t_{\text{bwd}}/t_{\text{fwd}}$ are reported in \ref{methods:sim_classical}; attack success rates are measured separately on the full test set using a GPU.

\paragraph{Results.}
Table~\ref{tab:classical_overhead} reports the measured gradient cost ratio and the quantum comparison. Two findings emerge.

\emph{First}, $\rho_{\text{classical}} \approx 5.2 \pm 0.6$ across all~$d$, with no systematic trend ($\rho \propto d^{-0.04}$, exponent indistinguishable from zero), confirming the $O(1)$ prediction of the Baur--Strassen theorem (Remark~\ref{rem:classical_cost}). \emph{Second}, the quantum overhead grows polynomially: $\rho_{\text{quantum}} \propto d^{3.00}$ (95\% CI $[2.91,3.08]$) on MNIST and $d^{3.07}$ (95\% CI $[2.99,3.15]$) on Fashion-MNIST, extracted from the in-regime shot-scaling fits of Section~\ref{subsec:shot_scaling}; the realized exponent reflects the tested models' measured gradient-norm decay, with the $d^{2.5}$ geometric floor applying to plateau-mitigated models with $\|\gstar\|=\Theta(1)$. The ratio $\rho_{\text{quantum}}/\rho_{\text{classical}}$ grows as $\Theta(d^{3.00})$ on MNIST and $\Theta(d^{3.07})$ on Fashion-MNIST, independent of the forward-inference normalization, and the quantum overhead exceeds the classical one for $d\gtrsim30$ (Table~\ref{tab:classical_overhead}).

\begin{table}[t]
\centering
\caption{\textbf{Gradient cost ratio: quantum vs.\ classical.} Classical $\rho$ (CPU wall-clock, median of 500 samples $\times$ 5 reps) is flat at $\approx 5$; quantum $\rho$ (polynomial fit of the shot-scaling data, Section~\ref{subsec:shot_scaling}) grows polynomially. $R_{\text{bwd}}=2d\,c_0(d)/\Delta L^\star$ is the shot count to reach a fixed attack-efficacy target $\Delta L^\star$ (\ref{app:classical_details}); $R_{\text{fwd}}=100$ is the forward-inference budget (Proposition~\ref{prop:backward_dominance}; $C=q$ commuting outputs at $\sigma_{\mathrm{fwd}}^2\approx3$, rounded up conservatively). The \emph{ratio} column scales as $1/R_{\text{fwd}}$; the \emph{exponent} is independent of $R_{\text{fwd}}$ and $\Delta L^\star$. For $d\le25$ the quantum overhead is at or below the classical value. Values to two significant figures.}
\label{tab:classical_overhead}
\small
\begin{tabular}{rrrrrr}
\toprule
& & \multicolumn{2}{c}{MNIST} & \multicolumn{2}{c}{Fashion-MNIST} \\
\cmidrule(lr){3-4}\cmidrule(lr){5-6}
$d$ & $\rho_{\text{cl}}$ & $\rho_{\text{qu}}$ & ratio & $\rho_{\text{qu}}$ & ratio \\
\midrule
16  & 5.4 & 0.38             & $0.07{\times}$ & 0.42             & $0.08{\times}$ \\
25  & 4.9 & 1.7              & $0.3{\times}$ & 2.1              & $0.4{\times}$ \\
64  & 6.1 & 21               & $3.5{\times}$ & 28               & $4.7{\times}$ \\
121 & 5.4 & $1.7\times10^{2}$ & $31{\times}$  & $2.2\times10^{2}$ & $41{\times}$ \\
256 & 5.4 & $1.4\times10^{3}$ & $250{\times}$ & $2.0\times10^{3}$ & $370{\times}$ \\
484 & 4.5 & $1.0\times10^{4}$ & $2.3\times10^{3}{\times}$ & $1.5\times10^{4}$ & $3.4\times10^{3}{\times}$ \\
784 & 4.4 & $5.0\times10^{4}$ & $1.1\times10^{4}{\times}$ & $7.7\times10^{4}$ & $1.7\times10^{4}{\times}$ \\
\bottomrule
\end{tabular}
\end{table}

Figure~\ref{fig:normalized_overhead} visualizes this contrast on log--log axes. The classical line is flat; the quantum curves rise with slopes matching the fitted exponents from Section~\ref{subsec:shot_scaling} and run steeper than the $\Theta(d^{2.5})$ reference, the gradient-norm excess over the idealized baseline (\S\ref{par:gnorm_baseline}).

\begin{figure}[t]
    \centering
    \includegraphics[width=0.5\linewidth]{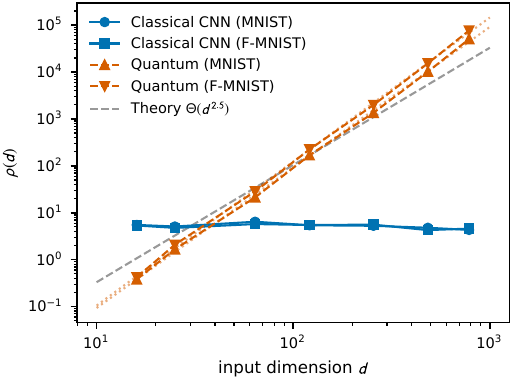}
    \caption{\textbf{Gradient cost ratio $\rho(d)$ on log--log axes.} Classical CNN (blue): flat at $\rho \approx 5$, confirming $O(1)$ scaling (Baur--Strassen). Quantum (red/orange): polynomial growth $\propto d^{3.00}$ (MNIST), $d^{3.07}$ (Fashion-MNIST). Dashed gray: $\Theta(d^{2.5})$ geometric-floor reference, anchored at the MNIST $d{=}121$ point. The quantum--classical overhead ratio grows as $\Theta(d^{3.00})$, with absolute value set by the forward-inference normalization (Table~\ref{tab:classical_overhead}).}
    \label{fig:normalized_overhead}
\end{figure}

\paragraph{Takeaway.}
Against the classical baseline, quantum measurement noise imposes a gradient cost ratio that is empirically $\rho_{\text{quantum}} \approx \Theta(d^{3})$ for the tested models (with $\Theta(d^{2.5})$ the geometric floor under $\|\gstar\|=\Theta(1)$) while classical backpropagation gives $\rho_{\text{classical}} = O(1)$ (the cheap-gradient principle). The contrast is structural, it arises from the decoupling of forward and backward measurement costs in quantum systems (Remark~\ref{rem:classical_cost}), and the overhead ratio grows with every additional input dimension.

\subsection{Hardware validation}
\label{sec:hw_validation}

The theoretical predictions and simulation results above establish shot noise as a fundamental impediment to gradient-based attacks. We now test whether this robustness survives deployment on real quantum hardware, where gate errors, decoherence, and readout imperfections introduce additional noise beyond idealized shot-sampling statistics.

The deployed model is a 4-qubit hardware-efficient classifier on $3\times 4$ MNIST ($d=12$), attacked with $\ell_2$ FGSM at $\epsilon=1.0$ across a logarithmic sweep of the attacker's shot budget $s$ on both an ideal simulator and \texttt{ibm\_boston}. The defender is evaluated noiselessly, so device noise enters only through the attacker's gradient; deployment, tiling, and sampling details are in \ref{methods:hw_protocol}.

\subsubsection{Results: shot-noise robustness on real hardware}

Figure~\ref{fig:hw_sim_comparison} presents adversarial accuracy (a) and the loss gap to the full-precision attack (b) versus shot budget for simulation and experiment; Table~\ref{tab:compare} summarises the endpoints.

\begin{table}[t]
    \centering
     \caption{\textbf{Shot-noise robustness, simulation vs.\ experiment ($d=12$, $N=100$, $\epsilon=1.0$).} Adversarial accuracy at the smallest budget, at a high budget, and at the exact-gradient limit. The $10\%$ floor is the cohort's intrinsic-robust fraction and the $s\to\infty$ limit of the simulation; at high budget the experiment stays a couple of percentage points above the simulation and does not reach this floor, its gradient device-biased. The low-shot gap is shot-noise--induced on both.}
    \begin{tabular}{cccc}
    \hline
         & Acc ($s=2$) & Acc ($s=2048$) & Exact floor ($s\rightarrow\infty$) \\
         \hline
         Simulation & 46.3\% & 12.4\% & 10.0\% \\
         Experiment & 42.2\% & 14.0\% & n/a \\
         \hline
    \end{tabular}
    \label{tab:compare}
\end{table}

\paragraph{Hardware tracks the simulator.}
In both simulation and experiment, adversarial accuracy rises sharply as the budget falls, from $\sim$42--46\% at the two-shot minimum toward the intrinsic-robust floor (Fig.~\ref{fig:hw_sim_comparison}a), driven entirely by the shot budget at fixed $\epsilon$ with the critical transition near $s\sim12$. The 100-input cohort resolves the two curves well below their separation. They cross near $s\sim12$: above it the simulation attack is stronger by $1$--$3$ percentage points, the experimental gradient running behind ideal once the shot budget is large enough to expose the device's gate and readout error; below it the experiment is marginally stronger, consistent with the within-ball criterion, under which a more isotropic low-shot direction is slightly likelier to find a boundary crossing in the near-random regime. Only the simulation reaches the exact-gradient floor ($10\%$) as $s\rightarrow\infty$; at high budget the experiment plateaus a couple of percentage points above the simulation, a residual the device cannot remove because its gradient stays biased. The loss gap to the exact-optimal attack, $\Delta L(s)=L(\delta^*)-L(\delta_s)$, closes toward zero in simulation and toward this residual in experiment (Fig.~\ref{fig:hw_sim_comparison}b), the optimization-gap decay of Proposition~\ref{prop:fgsm-optimization}.

\paragraph{The hardware gradient is faithful.}
At the full 8192-shot measurement the parameter-shift gradient on hardware aligns with the exact gradient at a cohort-median cosine of $0.98$ (mean $0.90$): the typical sample is recovered almost exactly, with a few low-cosine samples pulling the mean down and producing the high-budget residual above. Gate errors and decoherence barely rotate the gradient once shot noise is removed; what blunts the attack at low $s$ is the irreducible sampling uncertainty of the estimator, which isolates shot noise as the operative mechanism. As hardware fidelity improves, this device contribution shrinks while measurement uncertainty remains, so the shot-noise gap persists; and as the data dimension grows ($d>10$), the quadratic shot requirement renders high-quality gradient attacks rapidly prohibitive, especially for multi-step methods whose budgets accumulate over iterations.

\begin{figure}[t]
    \centering
    \includegraphics[width=\textwidth]{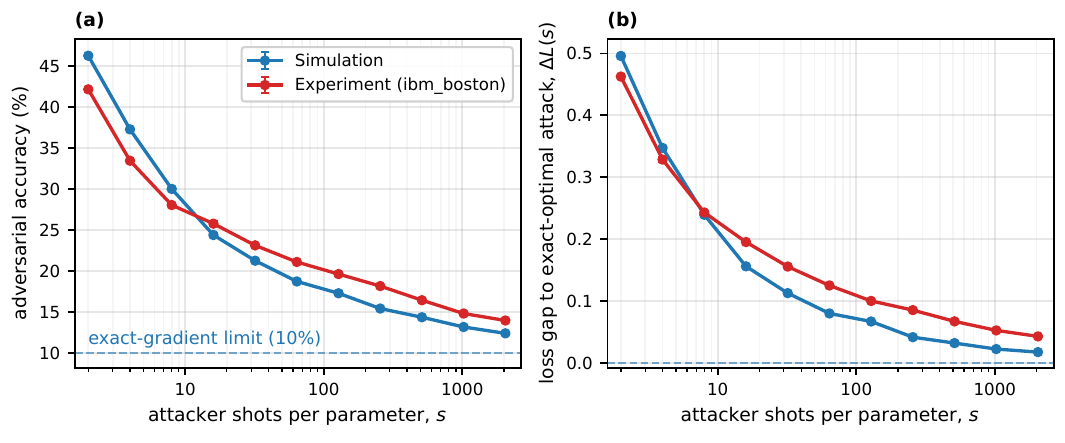}
    \caption{\textbf{Shot-noise robustness: simulation vs.\ experiment on \texttt{ibm\_boston}} ($d=12$, $N=100$, $\epsilon=1.0$). (\textbf{a})~Adversarial accuracy and (\textbf{b})~loss gap to the exact-optimal attack, $\Delta L(s)=L(\delta^*)-L(\delta_s)$, versus attacker per-parameter shot budget $s$ (log axis) under $\ell_2$ FGSM. The simulation reaches the exact-gradient floor ($10\%$, dashed) as $s\to\infty$; the experiment plateaus above it, its device-biased gradient bounded by the measured budget. Error bars are the standard error of the bootstrap mean over 200 shot-resamples (smaller than the markers). Protocol in the text and \ref{methods:hw_protocol}.}
    \label{fig:hw_sim_comparison}
\end{figure}

\section{Discussion and conclusion}
\label{sec:discussion}

The picture that emerges from theory and experiment is that finite quantum measurement imposes a shot-cost floor on gradient extraction, and that this floor reshapes adversarial attack efficacy as the input dimension grows. For single-step attacks we derive a simple law tying optimization error to the per-dimension shot budget, $\Delta L(s)\approx H(d)/s$ with $H(d)\propto \epsilon d$, which implies that maintaining fixed efficacy under growing dimension requires at least $s=\Theta(d)$ and thus a quadratic growth in total gradient shots. For multi-step C\&W, mapping the finite-shot updates to SGLD and applying Raginsky et al.'s non-asymptotic convergence theory yields a \emph{sufficient} total-shot budget polynomial in $d$ (\ref{app:sgld_sufficient}); this certificate provides a modeling framework that correctly predicts the observed bias--variance trade-off.

Experiments validate these predictions across model sizes, architectures, and estimators, and let us account for every measured exponent with no free parameters beyond the model's own properties. The sharpest test isolates the geometry of the bound: with the gradient norm folded out, the measured combination follows the parameter-free $d^{3/2}$ prediction ($d^{1.46}$ on MNIST, $d^{1.42}$ on Fashion-MNIST; \S\ref{par:gnorm_baseline}), and the waterfall plots (Fig.~\ref{fig:waterfall_mnist}) confirm the underlying $1/s$ decay. The bare per-dimension requirement (Fig.~\ref{fig:coeffic_vs_d}) is this geometry compounded by the model's gradient scale, and the exponent ledger closes explicitly: the geometric floor is $d^{3/2}$ per dimension (total $R=ds=d^{5/2}$); inserting the measured gradient-norm decay $\|\gstar\|\propto d^{-0.73}$ through the bound's $1/\|\gstar\|$ factor would predict $d^{3/2}\cdot d^{0.73}=d^{2.23}$ per dimension; and the measured landscape factor $\zeta(d)<1$, local curvature returning a small advantage to the attacker, pulls this back to the observed $d^{2.00}$ ($d^{2.07}$ on Fashion-MNIST), a total budget of $d^{3.00}$ ($d^{3.07}$). PSR-versus-FD comparisons confirm that the choice of estimator mainly modulates constants and low-shot transients, with PSR dominating at moderate or large budgets. In the multi-step regime, we observe a reproducible sweet spot $s^*(d)$ that drifts upward with dimension and can be improved by simple warm-up schedules; the critical budget $R_{\text{crit}}(d)$ rises across the four tested dimensions, though we do not fit a specific exponent given the limited range.

The resulting cost structure is a $d^{5/2}$ geometric shot-noise floor, realized as-is by plateau-mitigated models with $\|\gstar\|=\Theta(1)$~\cite{pesah_absence_2021,larocca_barren_2025}, and inflated to ${\approx}d^{3}$ for the tested generic deep circuits by the gradient-norm decay traced above; alternative estimators relocate the dimension cost instead of removing it (\ref{app:zero-order}), and for $d{=}784$ even a single FGSM sample already requires ${\sim}15$~hours of device time (Section~\ref{sec:Res}). The classical baseline comparison (Section~\ref{subsec:classical_baseline}) makes the overhead gap concrete. A matched classical CNN, trained on the same tasks with comparable parameter counts, yields $\rho_{\text{classical}} \approx 5$ across all dimensions: flat, exact, and dimension-independent, precisely as predicted by the Baur--Strassen theorem~\cite{baur_complexity_1983, griewank_evaluating_2008}. Against this anchor, the quantum gradient cost ratio of the tested models, $\rho_{\text{quantum}} \propto d^{p_1}$ with $p_1=3.00$ (95\% CI $[2.91,3.08]$) on MNIST and $p_1=3.07$ (95\% CI $[2.99,3.15]$) on Fashion-MNIST, makes the \emph{relative} gradient-extraction cost diverge as $\Theta(d^{p_1})$ against this dimension-independent baseline. Classical models are fully vulnerable at all dimensions (C\&W success $\ge 87\%$); quantum models are not intrinsically harder to fool. The defense is that the gradient information needed to fool them costs polynomially more to acquire. As scoped in Section~\ref{sec:scope}, this barrier is operative where the forward map is classically hard to simulate, the setting in which QML is expected to be useful, since only there is a white-box attacker forced onto measurement-based extraction instead of free classical backpropagation; the experiments here, run on deliberately simulable systems, establish the underlying scaling law that fixes that cost.

These scaling laws survive the move from simulation to real hardware. On IBM's 156-qubit \texttt{ibm\_boston} processor, a matched simulation-versus-experiment comparison at $d=12$ over a 100-input cohort reproduces the shot-noise robustness (Section~\ref{sec:hw_validation}): the high-shot gradient is faithful (cosine $0.90$), the simulation reaches the exact-gradient floor ($10\%$) as $s\to\infty$, and the experiment plateaus a few percent above it. That residual is a device bias the attacker cannot remove, so the observed low-shot robustness reflects finite measurement rather than a simulation artifact and persists under device imperfection. Improvements in gate fidelity and coherence will suppress this device-induced noise while leaving measurement uncertainty intact, sharpening the robustness gap rather than closing it.

Our shot-budget scaling laws are independent of structural robustness guarantees such as the dynamical conditions (unitarity, scrambling, chaos) studied in~\cite{dowling_adversarial_2024}: those conditions bound \emph{whether} certain attack classes are possible in principle, whereas our results quantify the \emph{physical measurement cost} of executing any gradient-based attack in practice. The two axes are orthogonal, and the shot-cost barrier applies even to attacks that such structural conditions permit. Two practical implications follow: evaluations of QML robustness should report and normalize by shot budgets, and defenses that amplify or exploit shot sensitivity (controlled noise injection, budget-aware training, adaptive shot allocation) are promising complements to architectural choices.

Several limitations bound these conclusions. We use fixed step sizes and hyperparameters, modest qubit counts $q$, and restricted repeats due to classical simulation cost; we focus on PSR as the primary estimator, though zero-order methods retain dimension-dependent shot requirements (\ref{app:zero-order}). The defensive consequence is scoped to the classically-hard regime (Section~\ref{sec:scope}): at the simulable sizes tested, a white-box attacker can simulate and backpropagate to circumvent the shot requirements entirely, so the experiments establish the scaling law rather than a deployed defense. The multi-step certificate of \ref{app:sgld_sufficient} is sufficient rather than necessary, and the empirical exponents exceed the $d^{5/2}$ baseline through the measured gradient-norm decay, with local curvature, measured directly, returning a small advantage to the attacker instead of penalizing it.

These limitations motivate future work on (i) realistic hardware noise beyond shot noise, (ii) principled annealing and adaptive shot-allocation schedules for multi-step attacks, (iii) model-level regularizers that widen the effective energy barriers (improving $\lambda_*$), (iv) extensions beyond gradient-based attacks, (v) adapting the finite-shot framework to QML regression with continuous outputs, (vi) full analysis of zero-order methods such as SPSA under shot constraints, including bias--variance--convergence trade-offs for adversarial optimization, and (vii) measuring the realized attack cost of plateau-mitigated architectures~\cite{pesah_absence_2021,Cerezo2021}, which the validated coefficient rule predicts sits at the $d^{5/2}$ geometric floor. Scaling toward $d{\sim}50$ to $100$ dimensions pushes single-step requirements into millions of shots, and shot-noise robustness can be further amplified via adversarial training under low-shot conditions or architectures that increase gradient variance. We hope these scaling laws and empirical patterns provide a clear target and a common language for benchmarking adversarial robustness in the finite-shot regime.

\ack
This work was supported by CSIRO under grant number~27040. We acknowledge the use of IBM Quantum services for this work; in particular, hardware experiments reported in Section~\ref{sec:hw_validation} were performed on the \texttt{ibm\_boston} processor accessed through the IBM Quantum Network. The views expressed are those of the authors and do not reflect the official policy or position of IBM or the IBM Quantum team. The authors declare no competing interests.

\section*{Data availability statement}
This study used publicly available datasets only (MNIST and Fashion-MNIST); no new training datasets were generated. The data that support the findings of this study, namely the aggregated quantum hardware run results, the low-shot bootstrap traces used in Section~\ref{sec:hw_validation}, and the per-dimension coefficient tables underlying Fig.~\ref{fig:coeffic_vs_d} and Fig.~\ref{fig:normalized_overhead}, will be made openly available in a public archive at a later date; in the interim they are available from the corresponding author on reasonable request. Raw IBM Quantum job records and per-circuit measurement strings are available from the corresponding author on reasonable request.
A cleaned and documented release of the simulation, attack, and analysis code that supports the findings of this study will be made openly available at a later date; in the interim, the code is available from the corresponding author on reasonable request.

\newpage
\appendix

\noindent The appendices are grouped in three parts. \emph{Methods and reproducibility} collects the simulation and hardware deployment protocols, the phase-dense encoding, and the classical-baseline and software details. \emph{Supporting theory} collects the multi-step SGLD mapping and sufficient-budget certificate, the covariance and cost-dominance derivations and proofs, the success-probability route, and the $\ell_2$-budget scaling. \emph{Robustness checks} collects the architectural and estimator variations, the supporting simulation figures, and the zero-order methods.

\section*{Part I: Methods and reproducibility}

\section{Simulation protocol}
\label{methods:sim_protocol}
This subsection collects the reproducibility-critical experimental configuration that supports the simulation results in Section~\ref{sec:Res} (single-step shot scaling, multi-step allocation, and the classical baseline comparison).

\subsection{Common experimental configuration}
\label{methods:sim_common}
All experiments draw from a shared base configuration so that individual subsections isolate one factor at a time while adhering to the setup discussed in Section~\ref{sec:background}.
\begin{itemize}
    \item \textbf{Datasets and encoding.} Unless noted otherwise, we use MNIST and Fashion–MNIST with phase–dense encoding (\ref{app:phase_dense}), mapping the downsampled input to $q$ qubits so that $d\leq2^q$ and $d = a^2$ for integer $a$ of pixel length of the square downsampled image. Resulting in $(q,d)\in\{(4,16),\allowbreak(5,25),\allowbreak(6,64),\allowbreak(7,121),\allowbreak(8,256),\allowbreak(9,484),\allowbreak(10,784)\}$; these dimensions span two orders of magnitude with systematic doubling, and $d=784$ is the native MNIST resolution ($28\times28$). The number of classes is $C=q$, choosing from the first $q$ classes of each dataset and balancing the training and test splits. Data is normalized to $[-\pi/2, \pi/2]$ before encoding so that the Gaussian concentration argument in \ref{app:gauss_l2} aligns with the $\epsilon=\kappa\sqrt{d}$ scaling used in Proposition~\ref{prop:fgsm-optimization}.
    \item \textbf{Models and training.} Baseline classifiers employ 200 strongly-entangling layers (SEL; Fig.~\ref{fig:ansatz_sim}) with identical initialization and optimization hyperparameters across qubit counts. Trainable rotation parameters are drawn entrywise from $\mathcal{N}(0,1)$; the default random-number-generator state was used and no fixed seed was set, and we verified that multiple independent re-runs produced consistent training trajectories and final test accuracy. Training uses Adam with cross-entropy loss; learning rate $0.001$ for $q\le 7$ and the variants used for $q\in\{8,9,10\}$ are documented in the released configuration files. Parameter counts are matched to the classical baseline of \S\ref{methods:sim_classical} so that the two pipelines compare like-for-like across $d$: each quantum classifier carries $200\times 3\times q$ trainable rotation parameters (e.g.\ $6{,}000$ at $q=10$), and the CNN's channel width $c=2q$ yields a comparable parameter budget at every dimension (Table~\ref{tab:classical_params} in \ref{app:classical_details}). When architectural variants are tested (\ref{app:arch_general}), only the ansatz changes; training schedules, data splits, and normalization remain fixed, preserving the Jacobian factors $J_L$ and $J_U$ that appear in the variance expression (\ref{eq:variance1}). A deep, generically-entangling stack of this kind generates a near-maximal dynamical Lie algebra (the highly expressive class for which the expressivity--measurement-efficiency trade-off~\cite{chinzei_tradeoff_2025} forces few simultaneously measurable gradient components), so the no-amortization premise of the single-step bound is met by construction at every $q$, independently of whether the tested instance is classically simulable.
    \item \textbf{Attack normalization.} For attacks, we set the perturbation radius to $\epsilon=\kappa \sqrt{d}$, matching the Gaussian norm concentration rationale in \ref{app:gauss_l2}. Gradients are estimated with the parameter–shift rule (PSR) using an average of $s$ shots per input dimension, giving total budget $R=ds$ per query. Loss surplus $\Delta L(s)=L(\delta)-L(\delta^*)$ and post–attack accuracy are the primary metrics, mirroring the optimization error in Proposition~\ref{prop:fgsm-optimization} and the SGLD framework of \ref{subsec:sgld_summary}.
    \item \textbf{Shot-grid provisioning.} For the single-step sweeps, per-dimension shot budgets are placed at up to eight log-spaced multiples $h=s/s^\star$ of the critical count $s^\star(d)=(d-1)\sigma^2/2\|\gstar\|^2$, spanning $h\in[0.3,12]$ ($d=784$ omits the lowest sub-critical multiples), so that every dimension reaches $h=12$, well inside the small-noise regime. The $\|\gstar\|$ and $\sigma^2$ values entering the grid construction come from preliminary per-sample probe runs used only to set the grids; they are not part of the reported attack curves. Coefficient fits use the window $h\ge1$, i.e.\ from the critical count $s^\star$ upward.
    \item \textbf{Iterative attacks.} When studying C\&W, we retain the same $\epsilon$ schedule and distribute a fixed total budget $R$ over $K$ iterations with per–dimension per–iteration shots $s=R/(Kd)$. Unless otherwise specified, we use $K=5$ steps and a constant learning rate; tuning overhead (e.g., binary search on the C\&W coefficient $c$) is excluded from $R$.
    \item \textbf{Evaluation protocol.} For each $(q,s)$ pair, we attack 200 balanced test samples (rounded as needed). Reported quantities average over all attacked samples without multiple gradient re-sampling to keep the classical cost tractable.
\end{itemize}

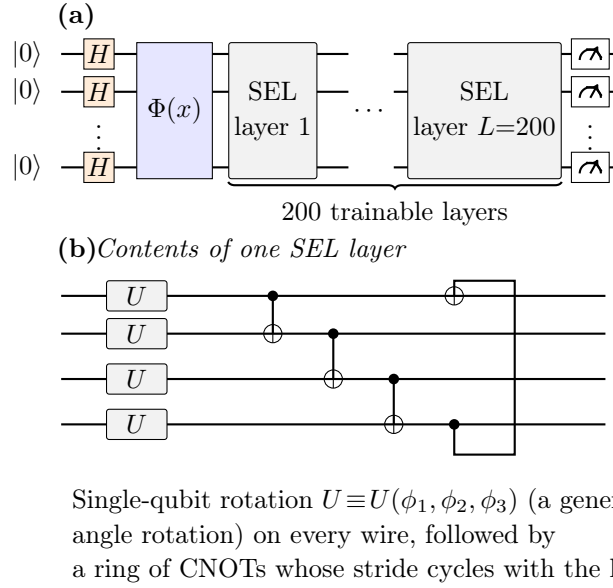
\begin{figure}[t]
\centering
\begin{tikzpicture}[
    every node/.style={font=\footnotesize},
    wire/.style={thick},
    block/.style={draw, rounded corners=1pt, fill=black!5, inner sep=2pt, font=\footnotesize, align=center},
    enc/.style={draw, fill=blue!10, inner sep=2pt, font=\footnotesize, align=center},
    had/.style={draw, fill=orange!15, minimum height=4mm, minimum width=4mm, inner sep=0pt, font=\footnotesize},
    meter/.style={draw, fill=white, minimum height=4mm, minimum width=5mm, inner sep=0pt},
    dots/.style={font=\footnotesize},
    brace/.style={decorate, decoration={brace, amplitude=2.5pt, mirror, raise=1pt}, thick},
    x=1mm, y=1mm,
]
  \def\meter#1#2{%
      \begin{scope}[shift={(#1,#2)}]
        \node[meter] (m) at (0,0) {};
        \draw[thick] (-1.4,-0.8) arc[start angle=180, end angle=0, radius=1.4];
        \draw[->, thick] (0,-0.8) -- (1.0,1.0);
      \end{scope}}

  \node[anchor=west, font=\footnotesize\bfseries] at (-2, 22) {(a)};
  \foreach \y in {2, 12, 17} {
    \draw[wire] (0, \y) -- (38, \y);
    \draw[wire] (44, \y) -- (74, \y);
  }
  \foreach \y in {2, 12, 17} { \node[font=\footnotesize, anchor=east] at (-1, \y) {$\ket{0}$}; }
  \foreach \y in {2, 12, 17} { \node[had] at (5, \y) {$H$}; }
  \node[font=\footnotesize] at (5, 7) {$\vdots$};
  \node[enc, minimum width=10mm, minimum height=18mm] (a-enc) at (15, 9.5) {$\Phi(x)$};
  \node[block, minimum width=10mm, minimum height=18mm] (a-sel1) at (28, 9.5) {SEL\\ layer 1};
  \node[font=\footnotesize] at (41, 9.5) {$\cdots$};
  \node[block, minimum width=18mm, minimum height=18mm] (a-selN) at (56, 9.5) {SEL\\ layer $L{=}200$};
  \draw[brace] (a-sel1.south west) -- (a-selN.south east)
      node[midway, below=4pt, font=\footnotesize, align=center]
      {200 trainable layers};
  \foreach \y in {2, 12, 17} { \meter{70}{\y} }
  \node[font=\footnotesize] at (70, 7) {$\vdots$};

  \begin{scope}[yshift=-34mm]
    \node[anchor=west, font=\footnotesize\bfseries] at (-2, 25) {(b)};
    \node[anchor=west, font=\footnotesize\itshape] at (3, 25) {Contents of one SEL layer};
    \foreach \y in {2, 8, 14, 19} { \draw[wire] (0, \y) -- (72, \y); }
    \foreach \y in {2, 8, 14, 19} {
      \node[block, minimum width=8mm, minimum height=4mm] at (10, \y) {$U$};
    }
    \fill (28, 19) circle (0.7);
    \draw[wire] (28, 19) -- (28, 14);
    \draw (28, 14) circle (1.2);
    \draw (28-1.2, 14) -- (28+1.2, 14);
    \draw (28, 14-1.2) -- (28, 14+1.2);
    \fill (36, 14) circle (0.7);
    \draw[wire] (36, 14) -- (36, 8);
    \draw (36, 8) circle (1.2);
    \draw (36-1.2, 8) -- (36+1.2, 8);
    \draw (36, 8-1.2) -- (36, 8+1.2);
    \fill (44, 8) circle (0.7);
    \draw[wire] (44, 8) -- (44, 2);
    \draw (44, 2) circle (1.2);
    \draw (44-1.2, 2) -- (44+1.2, 2);
    \draw (44, 2-1.2) -- (44, 2+1.2);
    \fill (52, 2) circle (0.7);
    \draw[wire] (52, 2) -- (52, -2) -- (60, -2) -- (60, 21) -- (52, 21) -- (52, 19);
    \draw (52, 19) circle (1.2);
    \draw (52-1.2, 19) -- (52+1.2, 19);
    \draw (52, 19-1.2) -- (52, 19+1.2);
    \node[font=\footnotesize, anchor=north west, align=left] at (0, -5)
        {Single-qubit rotation $U\!\equiv\!U(\phi_1,\phi_2,\phi_3)$ (a general $SU(2)$ Euler-\\
         angle rotation) on every wire, followed by\\
         a ring of CNOTs whose stride cycles with the layer index $\ell$.};
  \end{scope}
\end{tikzpicture}
\caption{\textbf{Simulation ansatz.}
(\textbf{a}) The simulation pipeline (\ref{methods:sim_protocol}): Hadamards on each of $q\!\in\!\{4,\dots,10\}$ wires, a phase-dense encoding $\Phi(x)$ of the $d$-dimensional input ($d$ up to $784$), 200 strongly-entangling layers (SEL), and single-qubit Pauli-$Z$ readouts.
(\textbf{b}) Contents of one SEL layer: a column of trainable single-qubit rotations $U\!\equiv\!U(\phi_1,\phi_2,\phi_3)$ followed by a ring of CNOTs whose stride cycles with $\ell$. Each layer carries an independent set of trainable angles.}
\label{fig:ansatz_sim}
\end{figure}

\paragraph{Scope and generality.} The design space of QML models, spanning circuit ans\"{a}tze, encoding schemes, depths, datasets, and attack configurations, is vast, and exhaustive coverage is computationally infeasible. Our experimental methodology, therefore, isolates individual factors (architecture, estimator, iteration count) while holding others fixed, enabling controlled comparisons that reveal how each factor modulates the underlying scaling laws. The consistent $1/s$ decay and superlinear growth of the coefficient observed across all tested configurations suggest that the shot-noise constraints are a fundamental feature of measurement-based gradient estimation. While future work may extend the parameter ranges or introduce additional circuit families, the methodological principle of varying one factor at a time against a shared baseline ensures that conclusions about scaling behavior remain interpretable and generalizable.

\subsection{Classical baseline: model and timing protocol}
\label{methods:sim_classical}
For each $(q,d,C)$ configuration, we train a VGG-style CNN with width scaled by qubit count ($c=2q$ base channels) so that parameter counts remain comparable to the quantum circuits (Table~\ref{tab:classical_params} in \ref{app:classical_details}). The perturbation budget matches the quantum experiments exactly: $\epsilon = \sqrt{d}/32$, giving a constant $\epsilon/\sqrt{d} = 0.03125$ across all~$d$.

The gradient cost ratio is defined as $\rho_{\text{classical}} = t_{\text{bwd}} / t_{\text{fwd}}$, with wall-clock timing a proxy for this dimensionless arithmetic-cost ratio. To obtain stable measurements that reflect the intrinsic arithmetic cost ratio (rather than GPU kernel-launch artifacts; see \ref{app:classical_details}), we time per-sample forward and backward passes on CPU with 100 warmup iterations followed by 500~samples $\times$ 5~repetitions, reporting the median. The backward pass computes the gradient by reverse-mode automatic differentiation with respect to the input, yielding the exact gradient an attacker would use. Attack success rates are measured separately via standard FGSM and C\&W attacks on the full test set using a GPU.

\section{Hardware experimental setup}
\label{methods:hw_protocol}

\begin{figure}[t]
\centering
\begin{tikzpicture}[
    every node/.style={font=\footnotesize},
    wire/.style={thick},
    block/.style={draw, rounded corners=1pt, fill=black!5, inner sep=2pt, font=\footnotesize, align=center},
    enc/.style={draw, fill=blue!10, inner sep=2pt, font=\footnotesize, align=center},
    meter/.style={draw, fill=white, minimum height=4mm, minimum width=5mm, inner sep=0pt},
    dots/.style={font=\footnotesize},
    brace/.style={decorate, decoration={brace, amplitude=2.5pt, mirror, raise=1pt}, thick},
    x=1mm, y=1mm,
]
  \def\meter#1#2{%
      \begin{scope}[shift={(#1,#2)}]
        \node[meter] (m) at (0,0) {};
        \draw[thick] (-1.4,-0.8) arc[start angle=180, end angle=0, radius=1.4];
        \draw[->, thick] (0,-0.8) -- (1.0,1.0);
      \end{scope}}

  \node[anchor=west, font=\footnotesize\bfseries] at (-2, 22) {(a)};
  \foreach \y in {2, 7, 12, 17} { \draw[wire] (0, \y) -- (84, \y); }
  \foreach \y in {2, 7, 12, 17} { \node[font=\footnotesize, anchor=east] at (-1, \y) {$\ket{0}$}; }
  \node[block, minimum width=9mm,  minimum height=18mm] (h-pre)  at (8.5,  9.5) {HWE\\ $\times\,3$};
  \node[enc,   minimum width=6mm,  minimum height=18mm] (h-rx1)  at (17,   9.5) {$RX$};
  \node[block, minimum width=9mm,  minimum height=18mm] (h-e0)   at (26,   9.5) {HWE\\ $\times\,3$};
  \node[enc,   minimum width=6mm,  minimum height=18mm] (h-rx2)  at (34.5, 9.5) {$RX$};
  \node[block, minimum width=9mm,  minimum height=18mm] (h-e1)   at (43,   9.5) {HWE\\ $\times\,3$};
  \node[enc,   minimum width=6mm,  minimum height=18mm] (h-rx3)  at (51.5, 9.5) {$RX$};
  \node[block, minimum width=13mm, minimum height=18mm] (h-post) at (63,   9.5) {HWE\\ $\times\,6$};
  \draw[decorate, decoration={brace, amplitude=2.5pt, raise=1pt}, thick]
       (h-rx1.north west) -- (h-rx3.north east)
       node[midway, above=3pt, font=\footnotesize] {3 $RX$ uploads ($3{\times}4{=}12$ features)};
  \draw[brace] (h-pre.south west) -- (h-post.south east)
      node[midway, below=4pt, font=\footnotesize] {15 HWE layers};
  \foreach \y in {7, 12} { \meter{78}{\y} }
  \node[font=\footnotesize, align=center, anchor=south] at (78, 19.6) {margin\\ readout};

  \begin{scope}[yshift=-34mm]
    \node[anchor=west, font=\footnotesize\bfseries] at (-2, 25) {(b)};
    \node[anchor=west, font=\footnotesize\itshape] at (3, 25) {Contents of one HWE layer};
    \foreach \y in {2, 8, 14, 19} { \draw[wire] (0, \y) -- (66, \y); }
    \foreach \y in {2, 8, 14, 19} {
      \node[block, minimum width=8mm, minimum height=4mm] at (8, \y) {$U$};
    }
    \fill (24, 19) circle (0.7);
    \fill (24, 14) circle (0.7);
    \draw[wire] (24, 19) -- (24, 14);
    \node[font=\scriptsize, anchor=west] at (25.5, 16.5) {$RZZ$};
    \fill (24, 8) circle (0.7);
    \fill (24, 2) circle (0.7);
    \draw[wire] (24, 8) -- (24, 2);
    \node[font=\scriptsize, anchor=west] at (25.5, 5) {$RZZ$};
    \fill (44, 14) circle (0.7);
    \fill (44, 8) circle (0.7);
    \draw[wire] (44, 14) -- (44, 8);
    \node[font=\scriptsize, anchor=west] at (45.5, 11) {$RZZ$};
    \foreach \y in {2, 8, 14, 19} {
      \node[block, minimum width=8mm, minimum height=4mm] at (60, \y) {$U$};
    }
    \node[font=\footnotesize, anchor=north west, align=left] at (0, -5)
        {Single-qubit rotation $U\!\equiv\!U(\phi_1,\phi_2,\phi_3)$ on every wire,\\
         a two-step brick of parametric $RZZ$ entanglers (even pairs,\\
         then odd pair), and a closing column of $U$ rotations.};
  \end{scope}
\end{tikzpicture}
\caption{\textbf{Hardware ansatz.}
(\textbf{a}) Hardware deployment on \texttt{ibm\_boston} (\ref{methods:hw_protocol}): a 4-qubit hardware-efficient ansatz (HWE) encoding $d\!=\!12$ features on $q\!=\!4$ qubits. Three $RX(x)$ feature-upload blocks (four features each, one per qubit) are interleaved with HWE entangling layers: three layers before the first upload, three after each upload, and three more closing the circuit, for 15 HWE layers in total. Only the two middle qubits are measured: the binary label is the margin $\langle Z_{q_1}\rangle-\langle Z_{q_2}\rangle$ between them, while the outer qubits are entangled but left unread. The reduced depth keeps the transpiled circuit within the device coherence budget.
(\textbf{b}) Contents of one HWE layer: a column of trainable single-qubit rotations $U\!\equiv\!U(\phi_1,\phi_2,\phi_3)$, a two-step brick of parametric $RZZ$ entanglers (even pairs, then odd pair), and a closing column of $U$ rotations. Each layer carries an independent set of trainable angles.}
\label{fig:ansatz_hw}
\end{figure}
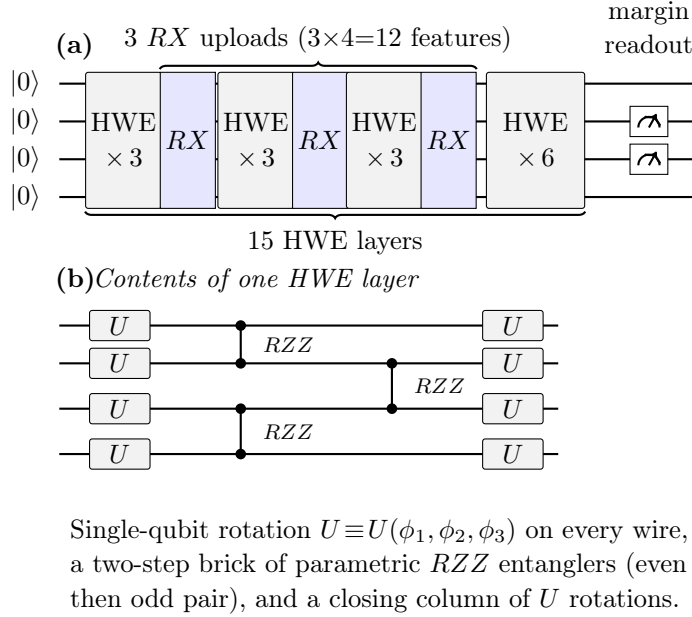

This subsection collects the deployment, attack-execution, and bootstrap-sampling protocol used for the hardware experiments reported in Section~\ref{sec:hw_validation}.

We deploy a 4-qubit hardware-efficient classifier (HWE; Fig.~\ref{fig:ansatz_hw}) on $3\times 4$ down-sampled MNIST ($d=12$), trained on an ideal noiseless simulator with the parameter-shift rule (PSR) and transferred to \texttt{ibm\_boston} without retraining; three $RX$ feature-upload blocks encode the $d=12$ features on the $q=4$ qubits. We retain HWE rather than the 200 strongly-entangling layers of \S\ref{methods:sim_protocol} so the transpiled depth stays within the device coherence budget.

We attack with $\ell_2$ FGSM at $\epsilon=1.0$ and sweep the attacker's per-parameter shot budget $s$ on a logarithmic grid; each parameter spends $s$ shots across its two PSR shift circuits (so $s\ge2$), and $R=d\cdot s$ is the total. The $2d{+}1$ PSR shift circuits run in parallel across 20 disjoint 4-qubit tiles (one input per tile), each measured once as a wide circuit at 8192 physical shots with dynamical decoupling and gate/measurement twirling, from which the budget-$s$ gradients are bootstrap-resampled.

From these counts we bootstrap-resample the per-shift measurements to synthesize the budget-$s$ gradient (each parameter using its two parameter-shift circuits, so the minimum budget is $s=2$). The simulation's $s\to\infty$ limit uses the exact gradient; on hardware the gradient retains the device bias, so the attack is bounded by the measured budget, short of the exact-gradient limit. The defender margin at the perturbed point is evaluated on an exact statevector, so device noise enters only through the attacker's gradient. A sample is attacked if any point on the $\epsilon$-ball segment along the noisy gradient direction is misclassified; we report the mean adversarial accuracy over 200 bootstrap resamples, with error bars giving its standard error (below $0.5\%$ at every budget). We compare \texttt{ibm\_boston} (median two-qubit error ${\sim}0.12\%$, $T_2{\sim}320\,\mu$s, readout ${\sim}0.35\%$, stable across the measurement window) against the ideal statevector simulator, on inputs pre-selected by correct noiseless classification (the full 100-input cohort, measured in five tile-batches of 20).

\section{Software stack and statistical fitting}
\label{methods:software}
Simulation and hardware experiments were run on Linux (Ubuntu via WSL) with Python~3.11. Quantum components used PennyLane~0.42 (with PennyLane-Lightning~0.42 and PennyLane-Qiskit~0.42), Qiskit~1.2, Qiskit-Aer~0.16 (Qiskit-Aer-GPU~0.15), Qiskit-IBM-Runtime~0.29, and Qiskit-Machine-Learning~0.8. Classical learning used PyTorch~2.10 (CUDA~12.8 build) with torchvision~0.25, alongside the standard scientific-Python stack. Power-law exponents are obtained by an unweighted log--log ordinary-least-squares fit of the per-dimension coefficient $c_0(d)$ (through-origin slope of the loss gap versus $1/s$ on the in-regime points $h\ge1$) against $d\in\{16,25,64,121,256,484,784\}$, yielding a slope $p_2$; the total-shot exponent reported in the body is $p_1=p_2+1$ (since $R=d\cdot s$ and $s\propto c_0(d)$). The 95\% confidence interval is the Student-$t$ interval on the regression slope with $n{-}2{=}5$ degrees of freedom (the $+1$ shift does not change its half-width).

\section{Phase-dense encoding details}
\label{app:phase_dense}
Standard angle encoding injects one feature per qubit per encoding layer via single-qubit rotations $RZ(x_j)$, yielding a capacity of at most $q$ features per layer. To encode $d$ features with only $q$ qubits in a single encoding step, we employ a \emph{phase-dense} encoding defined by
\[
  U_{\mathrm{enc}}(x) = \exp\!\Bigl(i\sum_{j=1}^{d} x_j\,|j\rangle\!\langle j|\Bigr)
  = \mathrm{diag}\!\bigl(e^{i\,x_1},\, \ldots,\, e^{i\,x_{d}},\,1,\ldots,1\bigr),
\]
where $\{|j\rangle\}$ denotes the $2^q$-dimensional computational basis. Each classical feature $x_j$ is assigned to a distinct basis state, giving an encoding capacity of $d \leq 2^q$. This matches the experimental configurations $(q,d) \in \{(4,16),\allowbreak (5,25),\allowbreak (6,64),\allowbreak (7,121),\allowbreak (8,256),\allowbreak (9,484),\allowbreak (10,784)\}$; when $d < 2^q$, unused basis-state phases are set to zero.

\paragraph{PSR compatibility.}
The generator associated with feature $x_j$ is the projector $G_j = |j\rangle\!\langle j|$, whose spectrum is $\{0, 1\}$ with spectral gap $\Delta G_j = 1$. Since each $G_j$ has exactly two distinct eigenvalues, the parameter-shift rule (Eq.~\ref{eq:psr_intro}) applies directly:
\[
  \frac{\partial\langle O\rangle}{\partial x_j}
  \;=\; \tfrac{1}{2}\bigl[\langle O(x + \tfrac{\pi}{2}\,e_j)\rangle
                         - \langle O(x - \tfrac{\pi}{2}\,e_j)\rangle\bigr],
\]
requiring two shifted circuit evaluations per input dimension.

\paragraph{Hardware depth caveat.}
Any decomposition of $U_{\mathrm{enc}}$ into a standard hardware gate set such as $\{\mathrm{CNOT},\, U_3\}$ requires $\Theta(2^q)$ gates, i.e., circuit depth exponential in~$q$. This is unavoidable: encoding $d \leq 2^q$ independent phases into computational-basis amplitudes demands at least $d$ parameterized gates. In our experiments, classical statevector simulation sidesteps this constraint entirely. On near-term hardware, the exponential depth limits phase-dense encoding to small~$q$; however, the shot-noise scaling laws derived in the paper depend only on the number of input features~$d$ and the PSR structure (two shifted evaluations per feature), not on the physical gate depth. Alternative hardware-friendly encodings that inject fewer than $2^q$ features per layer would simply require more encoding repetitions (as in the data re-uploading variant) to achieve the same capacity, without altering the per-dimension shot cost.

\paragraph{Preprocessing.}
MNIST and Fashion-MNIST images are downscaled via average pooling to $d = a^2$ pixels (with $a$ chosen so that $d \leq 2^q$), then normalized to $[-\pi/2,\,\pi/2]$ before encoding. This normalization ensures that the Gaussian concentration argument in \ref{app:gauss_l2} aligns with the $\epsilon = \kappa\sqrt{d}$ scaling used in Proposition~\ref{prop:fgsm-optimization}.

\paragraph{Data re-uploading variant.}
For the re-uploading architecture tested in \ref{app:arch_general}, the same phase-dense encoding layer appears at the input and again after 100 entangling layers within the 200-layer circuit, yielding two encoding repetitions. Each repetition re-encodes the full input $x$, so each feature influences the circuit output through two pathways, raising both the gradient norm and the per-shot readout variance; the two effects offset in the cost coefficient, while the PSR evaluation count per input dimension doubles (\ref{app:arch_general}).

\section{Classical baseline: experimental details}
\label{app:classical_details}

This appendix provides full methodology and supplementary data for the classical neural network baseline reported in Section~\ref{subsec:classical_baseline}.

\paragraph{Architecture.}
Table~\ref{tab:classical_params} lists the classical CNN configurations. The VGG-style architecture uses two convolutional stages (Conv--ReLU--MaxPool) followed by adaptive average pooling and a single linear classifier:
\[
\resizebox{\columnwidth}{!}{$
  \text{Conv}(1,c,3) \to \text{ReLU} \to \text{MaxPool}(2) \to \text{Conv}(c,2c,3) \to \text{ReLU} \to \text{MaxPool}(2) \to \text{GAP}(1) \to \text{Linear}(2c, C),
$}
\]
with $c = 2q$ (channel base scaled by qubit count). This keeps parameter counts in the same order of magnitude as the quantum re-upload circuits.

\begin{table}[t]
\centering
\caption{\textbf{Matched classical CNN configurations.} Channel base, parameter counts, and depth selected to match the quantum experiments at each input dimension.}
\label{tab:classical_params}
\small
\begin{tabular}{ccccc}
\toprule
$q$ & $d = n^2$ & channels $(c, 2c)$ & classical params & quantum params ($\approx 600q$) \\
\midrule
4  & 16  & 8, 16   & 1{,}316 & 2{,}400 \\
5  & 25  & 10, 20  & 2{,}025 & 3{,}000 \\
6  & 64  & 12, 24  & 2{,}886 & 3{,}600 \\
7  & 121 & 14, 28  & 3{,}899 & 4{,}200 \\
8  & 256 & 16, 32  & 5{,}064 & 4{,}800 \\
9  & 484 & 18, 36  & 6{,}381 & 5{,}400 \\
10 & 784 & 20, 40  & 7{,}850 & 6{,}000 \\
\bottomrule
\end{tabular}
\end{table}

\paragraph{CPU timing rationale.}
Wall-clock timing is a \emph{proxy} for the dimensionless gradient cost ratio $\rho = t_{\text{bwd}}/t_{\text{fwd}}$, a ratio of arithmetic cost rather than absolute throughput. GPU timing at batch\_size=1 is dominated by kernel-launch latency ($\sim$0.9\,ms each direction), collapsing $\rho$ to~$\approx 1$ and hiding the true arithmetic ratio. At batch\_size=200, forward parallelization ($\sim$0.008\,ms/sample) outpaces the sequential backward pass ($\sim$0.12\,ms/sample), inflating $\rho$ to~$\approx 15$. Neither reflects the $O(1)$ FLOP ratio. Per-sample CPU timing avoids these GPU scheduling artifacts and yields a stable $\rho \approx 5$, consistent with the Baur--Strassen bound of $\le 5{\times}$ for computing a single loss-to-input gradient~\cite{baur_complexity_1983, griewank_evaluating_2008}. Training and attacks still run on GPU; only the timing measurement uses CPU.

\paragraph{Quantum overhead extraction.}
The quantum $\rho$ values in Table~\ref{tab:classical_overhead} come from the experimental data, not from the $d^{3.00}$ fit. Backpropagation returns the classical attacker an \emph{exact} gradient; the quantum attacker buys gradient accuracy with shots, the fit $\Delta L\approx c_0(d)/s$ giving the shortfall of the $s$-shot attack from the perfect-gradient ideal. The quantum backward cost is the shot budget that brings the attack within a fixed tolerance $\Delta L^\star$ of that ideal, the efficacy classical backpropagation reaches at its $\le5{\times}$ inference overhead. Each qubit count's fit (Section~\ref{subsec:shot_scaling}) yields a per-dimension coefficient $c_0(d)$ (units: loss$\,\cdot\,$shots); reaching tolerance $\Delta L^\star$ per component costs $c_0(d)/\Delta L^\star$ shots, so with the Parameter-Shift Rule's two shifted evaluations per input dimension $R_{\text{bwd}} = 2d\,c_0(d)/\Delta L^\star$. We set $\Delta L^\star$ to one unit of attack loss, a fixed tolerance common to all $d$, at which the in-regime $1/s$ law of Section~\ref{subsec:shot_scaling} is evaluated; a looser tolerance lowers the quantum cost. The forward cost is $R_{\text{fwd}} = 100$. Both $\Delta L^\star$ and $R_{\text{fwd}}$ shift the quantum curve vertically and leave the scaling exponent unchanged.

\paragraph{Classical gradient cost ratio isolated.}
Figure~\ref{fig:classical_rho_only} shows the classical overhead across all 14 configurations with a constant fit $\bar{\rho} = 5.2 \pm 0.57$ and a power-law fit yielding $\rho \propto d^{-0.04}$; an exponent indistinguishable from zero.

\begin{figure}[h]
    \centering
    \includegraphics[width=0.5\linewidth]{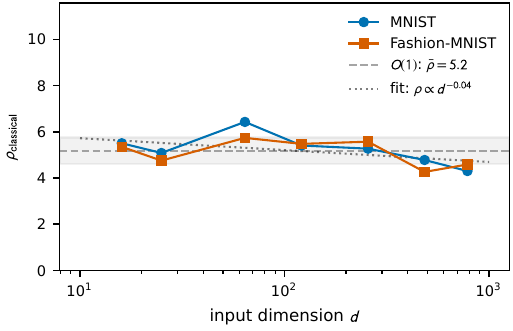}
    \caption{\textbf{Classical gradient cost ratio $\rho_{\text{classical}}(d)$.} Dashed gray: constant fit $\bar{\rho} = 5.2 \pm 0.57$. Dotted: power-law fit $\rho \propto d^{-0.04}$, confirming $O(1)$ scaling. Both MNIST (circles) and Fashion-MNIST (triangles) overlap, showing no dataset dependence.}
    \label{fig:classical_rho_only}
\end{figure}

\paragraph{Full attack success data.}
Table~\ref{tab:classical_attack_full} reports all 14 configurations.

\begin{table}[t]
\centering
\caption{\textbf{Classical attack success across dimensions.} Full results for FGSM and C\&W on both datasets and every input dimension.}
\label{tab:classical_attack_full}
\small
\begin{tabular}{clcccccc}
\toprule
$d$ & dataset & clean acc & FGSM & C\&W & $t_{\text{fwd}}$ (ms) & $t_{\text{bwd}}$ (ms) & $\rho$ \\
\midrule
16  & MNIST   & 92.0 & 50.1 & 100.0 & 0.039 & 0.225 & 5.51 \\
25  & MNIST   & 95.5 & 50.7 & 99.9  & 0.042 & 0.258 & 5.08 \\
64  & MNIST   & 97.9 & 21.4 & 99.9  & 0.043 & 0.248 & 6.43 \\
121 & MNIST   & 98.1 & 15.6 & 98.5  & 0.052 & 0.268 & 5.41 \\
256 & MNIST   & 96.4 & 22.1 & 89.4  & 0.051 & 0.266 & 5.27 \\
484 & MNIST   & 94.4 & 32.9 & 86.4  & 0.064 & 0.308 & 4.78 \\
784 & MNIST   & 93.0 & 54.7 & 87.0  & 0.077 & 0.330 & 4.30 \\
\midrule
16  & F-MNIST & 88.2 & 57.0 & 100.0 & 0.045 & 0.256 & 5.35 \\
25  & F-MNIST & 83.6 & 64.6 & 100.0 & 0.052 & 0.230 & 4.76 \\
64  & F-MNIST & 87.1 & 50.0 & 94.6  & 0.042 & 0.240 & 5.74 \\
121 & F-MNIST & 80.3 & 55.2 & 94.0  & 0.046 & 0.244 & 5.48 \\
256 & F-MNIST & 79.7 & 59.6 & 99.9  & 0.049 & 0.270 & 5.57 \\
484 & F-MNIST & 81.3 & 59.4 & 97.1  & 0.069 & 0.296 & 4.27 \\
784 & F-MNIST & 82.0 & 67.4 & 96.0  & 0.077 & 0.329 & 4.59 \\
\bottomrule
\end{tabular}
\end{table}

\section*{Part II: Supporting theory and derivations}

\section{Multi-step Carlini--Wagner attacks as stochastic gradient Langevin dynamics}
\label{app:sgld_sufficient}

This appendix gives the SGLD mapping summarized in Section~\ref{sec:multistep} and the sufficient-budget certificate it yields. The certificate is a \emph{sufficient condition} for guaranteed attacker convergence to a target error $\Delta$, derived from Raginsky et al.'s convergence theorem~\cite{raginsky_non-convex_2017}: it upper-bounds a budget that suffices, not one that is necessary, and so serves as a modeling framework consistent with the empirical evidence of Section~\ref{sec:Res} rather than as a lower bound on attacker cost.

\subsection{The C\&W--SGLD mapping}
\label{sec:CW}

The C\&W attack jointly minimizes misclassification loss and perturbation size,
\begin{equation}
\label{eq:CWloss}
    \underset{\delta}{\mathrm{minimize}} \quad c\,L_{mc}(x+\delta) + \|\delta\|_2^2,
\end{equation}
where $c$ balances the two objectives and is assumed well-tuned. Solving \eqref{eq:CWloss} by gradient descent with PSR-estimated gradients, the finite-shot update splits into the deterministic gradient and a measurement-noise term. Modeling the per-dimension shot noise as isotropic with the per-shot variance $\sigma^2$ of Section~\ref{sec:notation}, so that its covariance is $(\sigma^2/s)I$, the update reads
\begin{equation}
\label{eq:SGLD}
    \delta_{t+1} = \delta_t - \eta\,\nabla_\delta L(\delta_t) + \sqrt{2\eta/\beta}\,\xi,\qquad \xi\sim\mathcal N(0,I),
\end{equation}
which is exactly a gradient Langevin (SGLD) iteration once the injected-noise scale is matched to the shot noise. Equating the two noise scales gives the change of variable
\begin{equation}
\label{eq:ChangeOfVar}
    \beta = \frac{2s}{\eta\sigma^2},
\end{equation}
so the effective inverse temperature is linear in the per-iteration shot budget, $\beta\propto s$. A low budget is a high-temperature, exploratory trajectory; a high budget is a low-temperature, greedy descent. This coupling is the modeling framework used to read the iterative experiments (Section~\ref{subsec:multistep}).

\subsection{Sufficient-budget certificate}
\label{subsec:sgld_summary}

Raginsky et al.~\cite{raginsky_non-convex_2017,xu_global_2020} bound the expected optimality gap $\mathbb{E}[L(\delta_K)]-L(\delta^*)$ of the gradient Langevin iterate under three assumptions on the loss $\delta\mapsto L(\delta)$: $M$-smoothness, $(m,b)$-dissipativity, and an initialization law with finite $\kappa_0=\log\int_{\mathbb{R}^d} e^{\|\delta\|^2}p_0(\delta)\,d\delta$. The C\&W loss \eqref{eq:CWloss} with $p=2$ satisfies all three for a canonical PQC/softmax/C\&W model. The quantum layer is smooth because unitary encoding $e^{-iH_jx_j}$ has input derivatives bounded by the encoding-Hamiltonian norms $\|H_j\|$~\cite{kverne2025lipschitzsmooth} (the property that also underlies the trace-distance robustness guarantee of~\cite{dowling_adversarial_2024}), and the $\|\delta\|_2^2$ regularizer supplies dissipativity; the initialization condition holds for a point mass at $\delta=0$ or any $\mathcal N(0,\sigma_0^2 I_d)$ with $\sigma_0<1/\sqrt2$. Verifying these for the model yields $d$-dependent constants $M=O(cd+1)$, $b=O(c^2 d^2)$, $B=O(c)$, and gradient-noise variance $\sigma_g^2=O(c^2\sigma_\Sigma^2 d/s)$, with $m=1$ and $\kappa_0$ dimension-independent, where the bounded measurement-covariance condition $\|\Sigma_v\|\le\sigma_\Sigma^2$ keeps the observable norm $O(1)$ under unitary conjugation.

Combining the bound with the change of variable \eqref{eq:ChangeOfVar} translates the convergence requirement into a per-iteration condition $s\ge O((d/\Delta)\log(d/\Delta))\,\eta\sigma^2$, the theoretical analog of the empirical sweet spot $s^*(d)$, and, with the iteration count the same argument requires, a total budget
\begin{equation}
\label{eq:Rscale}
    R = Kds = O\!\left(\frac{d^{16}\log^{15}(d/\Delta)}{\Delta^{19}}\left(\frac{\log(d/\Delta)-\log\lambda'_*(d)}{\lambda'_*(d)}\right)^{5}\right)\eta\sigma^2,
\end{equation}
where $\lambda'_*(d)$ is the uniform spectral gap of the Langevin diffusion at the error-matching temperature. Two points make this informative despite being only sufficient. First, Raginsky's bound is monotone in each of $M,b,B,1/m$, and the verified QML constants above grow polynomially in $d$, so substituting them can only enlarge the budget: the exponent $16$ obtained under $d$-independent constants is a lower estimate of the certified degree. Second, the spectral gap satisfies $\lambda_*\gtrsim\mathrm{poly}(d,\beta)^{-1}e^{-\beta\Delta F}$ in multi-well landscapes, and this exponential-in-$\beta$ factor drives the high degree once $\beta$ is set to its error-matching value; for landscapes with favorable geometry $\lambda_*$ is far larger, which is why the empirical $R_{\mathrm{crit}}(d)$ of Section~\ref{subsec:multistep} sits far below the worst-case certificate.

The certificate is thus loose by construction and is not a lower bound on attacker cost; it establishes only that current finite-shot convergence theory grants an iterative attacker no guarantee of beating the single-step scaling within any feasible budget. What the mapping contributes instead is the $\beta\propto s$ modeling framework, which predicts the observed bias--variance trade-off and the sweet spot $s^*(d)$ (Section~\ref{subsec:multistep}).

\section{Supporting derivations and proofs}
\label{app:supp_derivations}

\subsection{Gradient-estimate covariance bookkeeping}
This expands the covariance structure summarized after Eq.~\eqref{eq:variance1}. With PSR variance $\mathbb V[\hat J_{ij}]=\Theta(\Delta_{G_j}^2/s_j)$, the commuting case (shared bitstrings at fixed shift index $j$, disjoint shot pools across different $j$) gives
\begin{align}
\mathrm{Cov}(\hat J_{ij},\hat J_{i'j'})&=\mathbb V[\hat J_{ij}] && \text{if}\ \ i'=i,\ j'=j, \\
\mathrm{Cov}(\hat J_{ij},\hat J_{i'j'})&\neq 0 && \text{if}\ \ j'=j,\ i'\neq i, \\
\mathrm{Cov}(\hat J_{ij},\hat J_{i'j'})&=0 && \text{if}\ \ j'\neq j.
\end{align}
and the non-commuting case makes all partial derivatives independent. Writing $V_{iv}=\mathbb V(\hat J_{iv})$, the full vectorized covariance of $\hat G=J_U\hat J J_L$ is
\[
\mathrm{Cov}\big[\mathrm{vec}(\hat G)\big]=(J_L^\top\!\otimes J_U)\,\mathrm{diag}(\mathrm{vec}(V))\,(J_L^\top\!\otimes J_U)^\top,
\]
which for $J_L=I_d$ reduces entrywise to the Hadamard-square weighting $\mathbb V[\hat G_{\cdot v}]=J_U^{\odot 2}\,V_{\cdot v}$.

\subsection{Proofs for the cost-dominance and correlated-noise results}

These complete the statements in Sections~\ref{methods:fp_bp_dominance} and~\ref{subsec:expected_relax}.

\begin{proof}[Proof of Proposition~\ref{prop:forward_variance}]
By Taylor expansion (Assumption~\ref{ass:fp_noise}), $J_U(\hat{y}_k) - J_U(y_k) \approx H_U(\hat{y}_k - y_k)$ to leading order. The readout error $\hat{y}_k - y_k$ is zero-mean Gaussian with covariance $\Sigma_y = \Theta(1/m)$ per observable (standard shot-noise scaling). For the Gaussian linear transformation,
\[
\mathbb{V}[H_U(\hat{y}_k - y_k)] = H_U\,\Sigma_y\,H_U^\top = \Theta(1/m).
\]
\end{proof}

\begin{proof}[Proof of Proposition~\ref{prop:backward_dominance}]
Forward-pass estimation requires $d_{\mathrm{out}}$ observable measurements, each needing $m = \Theta(1/\sigma_{\mathrm{fwd}}^2)$ shots to achieve target variance $\sigma_{\mathrm{fwd}}^2$. Total forward shots: $R_{\mathrm{fwd}} = \Theta(d_{\mathrm{out}}/\sigma_{\mathrm{fwd}}^2)$.

Backward-pass Jacobian estimation via PSR requires $2d_{\mathrm{in}}d_{\mathrm{out}}$ circuit evaluations (one for each partial derivative at two shifted parameter values), each needing $s = \Theta(1/\sigma_{\mathrm{bwd}}^2)$ shots. Total backward shots: $R_{\mathrm{bwd}} = \Theta(d_{\mathrm{in}}d_{\mathrm{out}}/\sigma_{\mathrm{bwd}}^2)$.

The ratio is
\[
\frac{R_{\mathrm{bwd}}}{R_{\mathrm{fwd}}} = \Theta\!\left(d_{\mathrm{in}}\,\frac{\sigma_{\mathrm{fwd}}^2}{\sigma_{\mathrm{bwd}}^2}\right).
\]
For comparable precision targets $\sigma_{\mathrm{fwd}}^2 = \Theta(\sigma_{\mathrm{bwd}}^2)$, backward cost dominates by a factor of $d_{\mathrm{in}}$.
\end{proof}

\begin{proof}[Proof of Corollary~\ref{cor:backward_focus}]
Immediate from Proposition~\ref{prop:backward_dominance}: as $d_{\mathrm{in}} \to \infty$ with fixed $d_{\mathrm{out}}$, the ratio $R_{\mathrm{bwd}}/R_{\mathrm{fwd}}$ diverges linearly.
\end{proof}

\begin{proof}[Proof of Corollary~\ref{cor:relaxed_shots}]
We relax the independence assumption and allow arbitrary correlations in the gradient error. The following derivation requires only that the error has a well-defined covariance matrix $\Sigma$ with finite trace; no Gaussianity or small-noise assumption is needed. For concreteness, write the one-shot gradient error model with covariance $\Sigma$ and average $s$ shots per input dimension:
\begin{equation}
    \hat{g} = \gstar + \xi,\qquad \xi\sim \mathcal{N}(\mathbf{0}, \Sigma/s).
\end{equation}
Using the deterministic geometry bound $\cos\alpha \ge (1-\rho)/(1+\rho) \ge 1-2\rho$ with $\rho=\|\xi\|/\|\gstar\|$ and Cauchy--Schwarz, we obtain the correlation-agnostic expected alignment guarantee
\begin{equation}
    \mathbb{E}[\cos\alpha] \;\ge\; 1 - 2\,\frac{\sqrt{\mathbb{E}[\|\xi\|^2]}}{\|\gstar\|}
    \;=\; 1 - 2\sqrt{\frac{\tr\Sigma}{s\,\|\gstar\|^2}}\;.
    \label{eq:expected_align_relax}
\end{equation}
Only $\Sigma\succeq 0$ with finite trace is required for this defensive effect to manifest; a one-line justification is $\mathbb{E}\,\|\xi\| \le \sqrt{\mathbb{E}\,\|\xi\|^2} = \sqrt{\tr\Sigma/s}$ by Cauchy--Schwarz, substituted into $\cos\alpha \ge 1-2\|\xi\|/\|\gstar\|$.

For single-step FGSM under Assumption~\ref{ass:1s_1}, the expected loss surplus relative to the ideal direction obeys
\begin{equation}
    \mathbb{E}[\Delta L] \;=\; \mathbb{E}\big[\epsilon\|\gstar\|(1-\cos\alpha)\big]
    \;\le\; 2\epsilon\,\sqrt{\frac{\tr\Sigma}{s}}\;.
    \label{eq:fgsm_surplus_relax}
\end{equation}
Since the right-hand side is decreasing in $s$, setting it equal to $\Delta$ and solving shows that the bound $\mathbb{E}[\Delta L]\le\Delta$ holds for every $s\ge s_\mathrm{min}(\Delta)=4\epsilon^2\tr\Sigma/\Delta^2$; thus $s_\mathrm{min}$ is the minimum sufficient per-dimension budget, and the dimension-explicit form follows from $\tr\Sigma=\Theta(d)$.
\end{proof}

\subsection{Single-step success-probability route}
This derivation underlies the tail guarantee stated at the end of Section~\ref{subsec:success_prob}; the geometry is shown in Fig.~\ref{fig:angle2robustness}.

First, we consider a one-step gradient-based attack that saturates the perturbation limit $\epsilon$ in one step by letting $\delta = -\epsilon \hat{g}$. Suppose that we consider the adversarial error rate, the percentage of incorrectly classified samples in the testing data, as the benchmark criterion for adversarial performance. Our performance will depend on both the model and the data structure. We shall first consider a single data point near the model boundary and measure the probability of the attack successfully perturbing the sample to cross the model boundary, as shown in Fig.\;\ref{fig:angle2robustness}. Here, we adopt the assumption \ref{ass:1s_1} and \ref{ass:1s_2} from the previous section and denote the distance from a sample to the linear boundary as $l_x$. Under these assumptions, and conditioning on an attackable sample with $l_x \le \epsilon$ (otherwise no perturbation of budget $\epsilon$ can reach the boundary), we notice that only if the angle $\alpha$ between the estimated gradient and the true gradient is smaller than $\tau = \arccos(l_x/\epsilon)$, the perturbed sample can cross the boundary and become an adversarial sample.

To further evaluate the robustness caused by the uncertainty of the gradient measurement during attacks, we consider the probability that the measurement-based attack $\delta = -\epsilon \hat{g}$ is successful. A successful attack will perturb the sample across the nearest model boundary, that is, $\alpha \leq \tau$. As established above, the normalized gradient estimate $\hat{g}_{n}$ follows a von Mises-Fisher distribution with concentration parameter $\kappa_{\mathrm{vMF}} \approx s \|\gstar\|^2/\sigma^2$ (cf.\ Equation~\ref{eq:CosSim}). Thus, the probability of the attack being successful $P(\|\hat{g}_n-\gstar_n\|<\tau)$ with some threshold $\tau$ is shown to be
\begin{equation}
    P(\|\hat{g}_n-\gstar_n\| < \tau) = \int_0^{\tau} f_r(r)\; dr = \frac{\gamma\left(\frac{d-1}{2}, \frac{s \tau^2 \|\gstar\|^2}{2\sigma^2}\right)}{\Gamma\left(\frac{d-1}{2}\right)}\;,
\end{equation}
where $\gamma(\cdot,\;\cdot)$ is the incomplete Gamma function and $\Gamma(\cdot)$ is the Gamma function \cite{mardia_tests_1999}.

If we assume that we are processing high-dimensional data so that $d$ is large, the probability will be further simplified via vMF angular concentration of the normalized gradient estimate (and the small-angle approximation $1-\cos\tau\approx\tau^2/2$) to
\begin{equation}
    P(\|\hat{g}_n-\gstar_n\| < \tau) \approx \Phi\!\left(\frac{\sqrt{2}\!\left(\sqrt{s}\,\tau\|\gstar\| - \sigma\sqrt{d}\right)}{\sigma}\right)\;,
\end{equation}
where $\Phi(\cdot)$ is the CDF of the standard normal distribution and the factor $\sqrt{2}$ arises from the CLT approximation of $\mathrm{Gamma}((d{-}1)/2,\,1)$, whose standard deviation is $\sqrt{(d{-}1)/2}$. The result reveals a key resource-dimension trade-off: to maintain the same attack success probability as dimension $d$ increases, the attacker must proportionally increase the shots per dimension $s$.

An important implication of this result is that maintaining a fixed attack success probability as the dimension $d$ increases requires a proportional increase in the number of shots per input dimension $s$.
\begin{equation}
    \sqrt{s}\,\tau\|\gstar\| - \sqrt{d}\,\sigma = \text{Constant}\;.
\end{equation}
Consider the success rate threshold of 50\%. We obtain the following new success-rate bound.

\begin{proposition}[FGSM shot budget for 50\% success]\label{prop:fgsm-success}
Under Assumptions~\ref{ass:1s_1}, \ref{ass:1s_2}, and~\ref{ass:1s_3}, with $\|\gstar\|=\Theta(1)$ (constant relative geometry), consider a sample $x$ in $d$ dimensions. To maintain a 50\% success rate under an $\ell_2$ FGSM attack, the shots per dimension must scale as
    \begin{equation}
    \label{eq:m}
    s = O\!\left(\frac{d}{\tau^2}\right) \;,
    \end{equation}
where the tolerance angle satisfies $\cos\tau=l_x/\epsilon$ (with $l_x \le \epsilon$ for the sample to be attackable) and $l_x$ is the shortest distance from $x$ to the decision boundary.
\end{proposition}

Now, we combine (\ref{eq:m}) and the assumptions that $\sigma$ and $\tau$ do not change significantly for data with higher dimensions. We found that the shots per dimension $s = O(d)$, and the total number of shots $R$ for all dimensions scale quadratically with the dimension $d$, i.e., $R=O(d^2)$ in the most optimistic attacker scenario (constant relative geometry). For future quantum models, this quadratic overhead can become an important limiting factor for measurement-based attacks, making reliable boundary crossing increasingly resource-intensive in high dimensions.

\section{Dimension scaling of the \texorpdfstring{$\ell_2$}{l2} perturbation budget}
\label{app:gauss_l2}
This appendix justifies the convention $\epsilon=\kappa\sqrt{d}$ used throughout the scaling experiments for the adversarial $\ell_2$ budget (the single-dimension hardware demonstration of Section~\ref{sec:hw_validation} fixes $\epsilon$ at one budget instead). The scaling has two complementary readings. (i) \emph{Per-coordinate imperceptibility.} A uniform per-coordinate bound $|\delta_i|\le\kappa$ directly implies $\|\delta\|_2\le\kappa\sqrt{d}$, so $\epsilon=\kappa\sqrt{d}$ is the $\ell_2$ envelope of any attack whose per-pixel magnitude is capped at the imperceptibility scale $\kappa$. (ii) \emph{Detectability-matched scale.} An adversarial perturbation whose $\ell_2$ norm matches the typical $\ell_2$ norm of natural Gaussian input noise at scale $\sigma=\kappa$ is statistically indistinguishable from benign noise under an $\ell_2$ detector, so $\epsilon=\kappa\sqrt{d}$ is the largest budget that remains below the natural-noise floor. The second reading is quantified by the standard $\chi^2_d$ concentration, which we include for completeness.

Assume image (or feature) noise is i.i.d.\ Gaussian across dimensions: $\eta\sim\mathcal{N}(0,\sigma^2 I_d)$. Then $\|\eta\|_2^2/\sigma^2\sim\chi^2_d$. Standard concentration for the chi-square law yields, for all $t\ge 0$ (see, e.g., \cite{laurent_massart_2000,vershynin_HDP_2018}):
\begin{align}
\Pr\big(\|\eta\|_2 \ge \sigma(\sqrt{d}+\sqrt{2t})\big) &\le e^{-t},\\
\Pr\big(\|\eta\|_2 \le \sigma(\sqrt{d}-\sqrt{2t})\big) &\le e^{-t}.
\end{align}
Equivalently, with probability at least $1-\delta$, the noise $\ell_2$ norm lies in the band
\[
\sigma\big(\sqrt{d}-\sqrt{2\log(1/\delta)}\big)\;\le\;\|\eta\|_2\;\le\;\sigma\big(\sqrt{d}+\sqrt{2\log(1/\delta)}\big).
\]
Thus, an $\ell_2$ perturbation whose magnitude scales as $\epsilon=\kappa\sqrt{d}$ matches the typical Gaussian growth and keeps a fixed tail probability when $\kappa$ is chosen relative to $\sigma$ (e.g., take $\kappa=\sigma$ for median scale, or $\kappa=\sigma\big(1+\sqrt{2\log(1/\delta)/d}\big)$ for a $(1-\delta)$ upper-quantile). In particular, thresholds of the form
\[
\epsilon_{\text{flag}}(\delta) = \sigma\big(\sqrt{d}+\sqrt{2\log(1/\delta)}\big)
\]
flag $\ell_2$ noise levels that are unlikely under $\mathcal{N}(0,\sigma^2 I_d)$ at significance $\delta$. This justifies the concise heuristic $\epsilon=\kappa\sqrt{d}$ in the main text and provides a principled way to tie $\kappa$ to $\sigma$ and a desired false-alarm rate $\delta$.

\section*{Part III: Robustness checks}

\section{Supporting figures for the simulation results}
\label{app:supp_figs}
This appendix collects the secondary figures and the landscape-curvature probe supporting the single-step and multi-step simulation results of Section~\ref{sec:Res}: the gradient-norm-folded coefficient, the local-linearity (landscape-factor) measurement, and the multi-step C\&W shot-budget sweep. The Fashion-MNIST shot-scaling waterfall is now shown alongside MNIST in Fig.~\ref{fig:waterfall_mnist} (panels \textbf{c}, \textbf{d}).

\begin{figure}[t]
    \centering
    \includegraphics[width=0.66\linewidth]{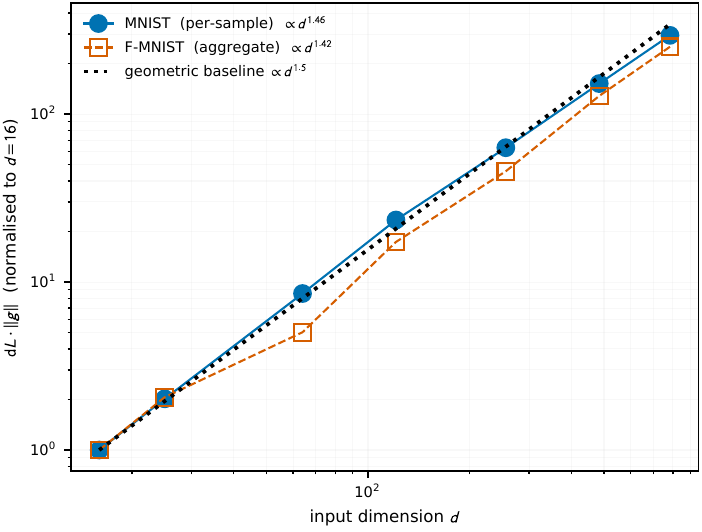}
    \caption{\textbf{Gradient-norm-independent coefficient.} The combination $\mathrm{d}L\cdot\|\gstar\|$ of Eq.~\eqref{eq:dLg} (the $1/s$ loss-gap coefficient with the measured gradient norm folded back in) versus input dimension $d$, normalised to $d{=}16$, against the geometric $d^{3/2}$ baseline (dotted). MNIST (per-sample, line-scan $\zeta_i$): $d^{1.46}$; Fashion-MNIST (aggregate $B_q\cdot\langle\|\gstar\|\rangle$): $d^{1.42}$. Folding $\|\gstar\|$ back into the coefficient recovers the shot-noise baseline on both datasets, leaving only the mild landscape factor $\zeta(d)$.}
    \label{fig:coeff_gnorm}
\end{figure}

\paragraph{Probing local linearity at the attack scale.} Assumption~\ref{ass:1s_1} treats the loss as linear over the $\epsilon$-step; the landscape factor quantifies the departure. For each input we step a distance $\epsilon$ along $\hat{d}(\alpha)=\cos\alpha\,\hat{u}^*+\sin\alpha\,\hat{t}$, rotating from the gradient direction $\hat{u}^*$ ($\alpha=0$) toward a random orthogonal direction $\hat{t}$, and record the excess over the first-order prediction, $\Delta(\alpha)=L(x+\epsilon\hat{d}(\alpha))-L(x)-\epsilon\|\gstar\|\cos\alpha$; here $L$ is the exact (noiseless) loss, so the scan isolates landscape geometry rather than shot noise. We sweep $\alpha$ over a uniform grid on $[0,75^\circ]$ and average $\Delta(\alpha)$ over eight random transverse directions $\hat{t}$ per sample. A linear landscape gives $\Delta\equiv0$, and the on-axis value defines $\zeta=1+\Delta(0)/\epsilon\|\gstar\|$ per sample. Measured at every dimension (Fig.~\ref{fig:nonlinearity}), $\zeta$ holds near unity through $d=121$ and falls monotonically to $0.83$ at $d=784$. This high-$d$ curvature favors the attacker: the loss is concave along the gradient at the $\epsilon$-scale and flatter transversely (the deficit $\Delta(\alpha)$ is most negative on-axis and relaxes toward zero off-axis), so a shot-noisy attack that tilts off the true gradient lands in a less-saturated direction and loses less loss-rise than the linear prediction. Curvature thus broadens the cone of effective attack directions and lowers the shot-noise penalty ($\zeta<1$), a small, dimension-growing advantage rather than a cost.

\begin{figure}[t]
    \centering
    \includegraphics[width=\textwidth]{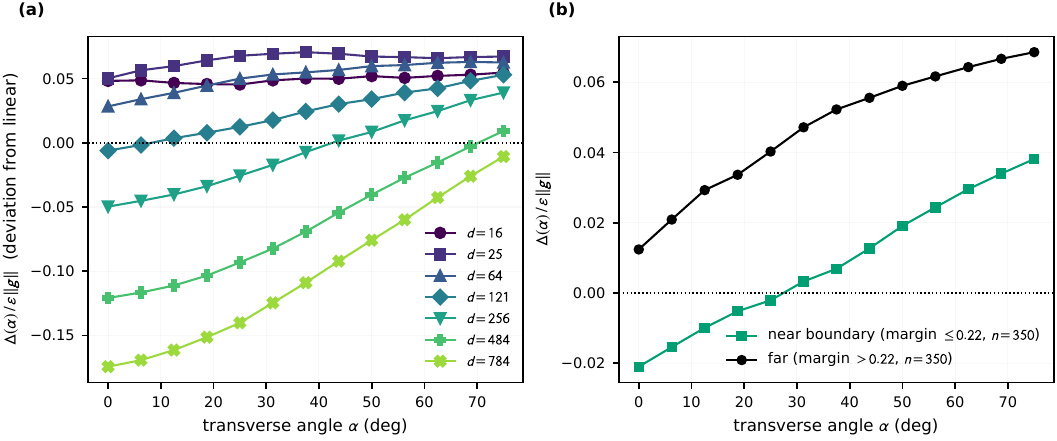}
    \caption{\textbf{Landscape non-linearity at the attack scale.} (\textbf{a}) Median transverse line-scan $\Delta(\alpha)/\epsilon\|\gstar\|$ versus tilt angle $\alpha$ for $d=16\to784$: the deviation from the first-order prediction is flat at low $d$ and grows strongly negative along the attack direction as $d$ increases; the $\alpha{=}0$ intercept equals $\zeta-1$ ($+0.05\to-0.17$). (\textbf{b}) Splitting samples by distance to the decision boundary at $d=121$: far-from-boundary samples carry the larger transverse deviation; the split washes out by $d=784$, where the non-linearity is a global property of the landscape.}
    \label{fig:nonlinearity}
\end{figure}

\begin{figure}[t]
    \centering
    \includegraphics[width=0.62\linewidth]{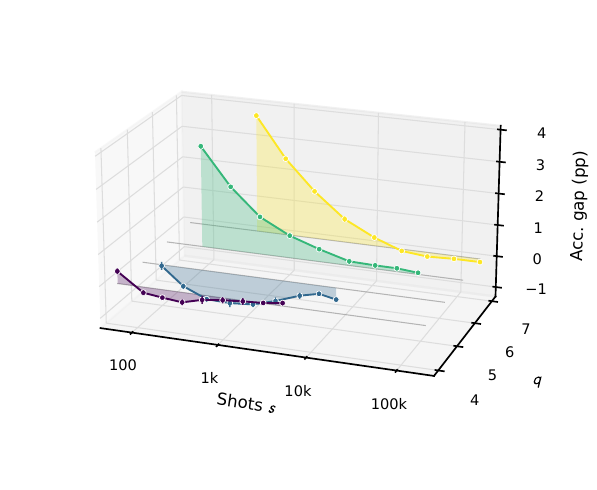}
    \caption{\textbf{C\&W shot-budget sweet spot.} Signed accuracy gap of the finite-shot attack to the exact-shot ($s\to\infty$) attack (percentage points, pp) versus shots per dimension $s$ (log scale) and model size $q$; the grey line is the exact-shot reference. Below it, finite shots make the attack \emph{stronger} than an exact gradient (the SGLD sweet spot; see text). Markers are means over five shot-noise realizations on the full test set; bars are $\pm1$ SEM.}
    \label{fig:waterfall_cw_asr}
\end{figure}

\paragraph{Aggregate bracket on the single-step exponent.} The directly fitted coefficient $d^{2.00}$ (MNIST) factors as the per-sample geometry, $d^{1.46}$, times the gradient norm, which enters with effective exponent $0.54$. The aggregate reference $d^{2.23}$ folds the cohort-mean inverse norm $1/\langle\|\gstar\|\rangle\propto d^{0.73}$ into the bare $d^{3/2}$ geometry and overstates the norm contribution: $\sigma^2$ and $\|\gstar\|$ are positively correlated across samples ($+0.75$), so low-norm samples carry low variance and the per-sample ratio $\sigma^2/\|\gstar\|$ grows more slowly than $1/\langle\|\gstar\|\rangle$. The $0.23$ separation is this norm-exponent damping ($0.73\to0.54$, contributing $0.19$) plus the $0.04$ offset between the bare $d^{3/2}$ and the measured $d^{1.46}$ geometry ($\zeta$ and finite-window extraction), so $d^{2.23}$ brackets the per-sample law $d^{2.00}$ from above, as an aggregated estimate should.

\section{Robustness of a different architecture: data re-uploading}
\label{app:arch_general}

This appendix provides a detailed experimental comparison that supports the architectural robustness paragraph in Section~\ref{subsec:shot_scaling}. We repeat the $(q,s)$ sweep from Section~\ref{subsec:shot_scaling}, under the same grid protocol, for a data re-uploading variant of that circuit, which re-encodes the input at mid-depth rather than encoding once at the input: the angle-encoding layer appears at the input and again after 100 entangling layers, two encoding repetitions within a fixed 200-layer depth. We fix the dataset to MNIST and vary $q\in\{4,5,6,7\}$, so that the only experimental difference is the circuit structure itself, isolating how the architecture-dependent Jacobian factors $J_L$ and $J_U$ in~(\ref{eq:variance1}) shift variance constants. The re-uploading model's own measured gradient norm ($\|\gstar\|\propto d^{-0.81}$, $1.5$--$1.8\times$ the single-upload values) and readout variance ($\sigma^2\approx4.7$, $1.4$--$1.8\times$ the single-upload $\sigma^2\approx3$, consistent with variance accumulating over the two encoding blocks) set its critical counts $s^\star(d)$, so the relative budgets $h=s/s^\star\in[0.3,12]$ are matched between the two architectures.

\begin{figure}[t]
    \centering
    \includegraphics[width=\textwidth]{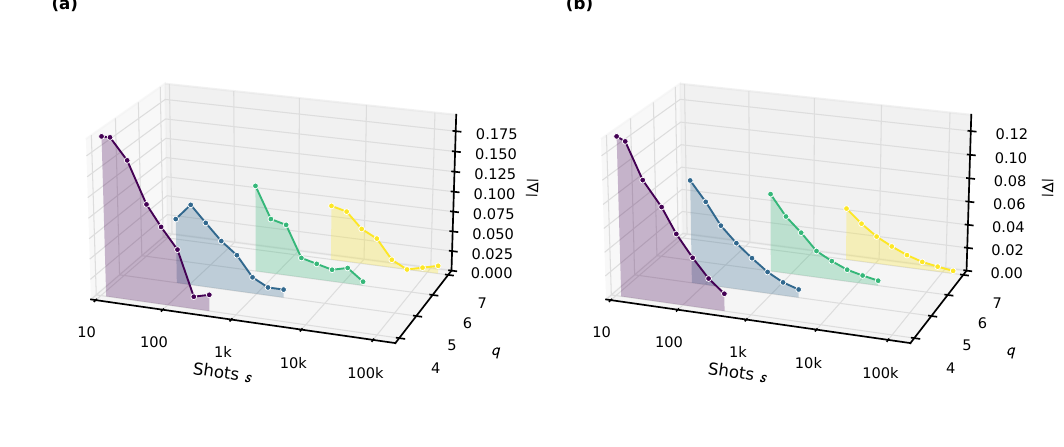}
    \caption{\textbf{Re-uploading ansatz: FGSM waterfall on MNIST.} Same convention as Fig.~\ref{fig:waterfall_mnist}: (\textbf{a}) accuracy and (\textbf{b}) loss shortfall from the perfect-gradient attack versus shot budget (log scale) for $q\in\{4,5,6,7\}$. The re-uploading architecture qualitatively mirrors the single-upload baseline, confirming the $1/s$ decay persists across circuit families.}
    \label{fig:waterfall_reupload}
\end{figure}

Figure~\ref{fig:waterfall_reupload} presents the re-uploading curves alongside Fig.~\ref{fig:waterfall_mnist} for direct comparison. The qualitative behavior is strikingly similar: both architectures exhibit the characteristic $1/s$ decay within each qubit count, precisely the signature predicted by Proposition~\ref{prop:fgsm-optimization}.

To quantify this similarity, we apply the same in-regime through-origin fitting procedure from Section~\ref{subsec:shot_scaling}. Figure~\ref{fig:Kd_compare} overlays the fitted coefficients $c_0$ for both architectures on log--log axes over the matched range $d\in\{16,25,64,121\}$. Two key observations emerge. First, the \emph{exponents} are statistically indistinguishable: $d^{1.99}$ ($R^2=0.991$) for re-uploading versus $d^{1.96}$ ($R^2=0.997$) for single-upload on the matching rungs. Second, the \emph{coefficients} coincide (per-dimension pre-factor ratio $0.89$--$1.09$, geometric mean $0.98$): the second encoding pathway raises the gradient norm and the per-shot readout variance by nearly identical factors ($\approx1.6\times$ each), and the two cancel in the coefficient rule $c_0\propto\sigma^2/\|\gstar\|$. The architectural cost surfaces elsewhere: each input component appears in both encoding blocks, so one PSR gradient consumes four shifted circuit evaluations per dimension instead of two, doubling the total shot budget per gradient at equal $s$ without touching the coefficient or the exponent.

\begin{figure}[t]
    \centering
    \includegraphics[width=0.5\linewidth]{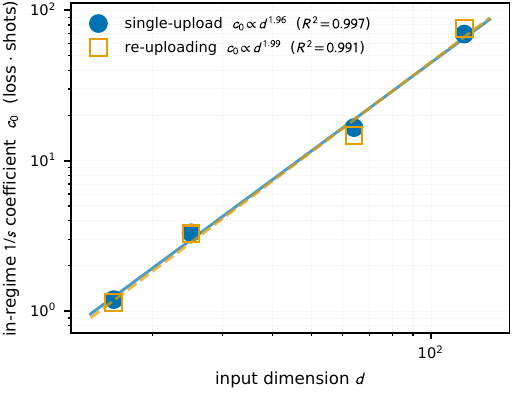}
    \caption{\textbf{Architectural robustness: re-uploading vs.\ single-upload coefficient scaling.} Log--log plot of the in-regime coefficients $c_0$ versus input dimension $d$ over the matched range $d\in\{16,25,64,121\}$, comparing re-uploading (orange, $\propto d^{1.99}$) to single-upload (blue, $\propto d^{1.96}$; cf. Fig.~\ref{fig:coeffic_vs_d}). The exponents are statistically indistinguishable and the pre-factors coincide (geometric-mean ratio $0.98$): the re-uploading circuit's larger readout variance is offset by its larger gradient norm in $c_0\propto\sigma^2/\|\gstar\|$. Lines show least-squares power-law fits.}
    \label{fig:Kd_compare}
\end{figure}

These results directly corroborate Section~\ref{subsec:shot_scaling}: the super-linear coefficient growth measured on single-upload MNIST reappears on the re-uploading variant within fitting uncertainty, and together with the Fashion-MNIST exponent (Section~\ref{subsec:shot_scaling}) we have three independent configurations all exhibiting growth far above the $d^{1.5}$ theoretical floor. The coefficient coincidence is itself informative: it confirms the structure $c_0\propto\sigma^2/\|\gstar\|$ on an independent architecture, since both measured ingredients moved by a common factor and the coefficient did not move. The consistency across datasets and architectures strongly suggests that the scaling law is an intrinsic property of PSR-based gradient estimation on quantum circuits. For practitioners, this implies that more expressive circuit families may multiply the per-gradient evaluation count, but the polynomial explosion in total budget with dimension persists regardless of ansatz.

\section{Estimator comparison: parameter-shift versus finite differences}
\label{app:estimator}

We assess how the choice of gradient estimator interacts with the shot budget and the variance model in~(\ref{eq:variance1}). Building on the baseline architecture at $q=6$ qubits, we contrast PSR, which is unbiased but potentially higher-variance at low $s$, with a two-point finite-difference (FD) estimator with step size $h$. In FD, we treat the quantum layer as a black box and choose $h\ll 1$ to avoid parameter clipping from periodic gate ranges. We enforce identical total shot budgets across methods so that variance-bias trade-offs, rather than raw sampling volume, drive the outcome; analytically, this isolates the shift from the unbiased PSR covariance in Assumption~\ref{ass:1s_2} to the biased estimator that perturbs the mean term $\mathbb E[\hat G]=J_U\,\mu\,J_L$ from Section~\ref{sec:uncertain_gradient}.

To visualize these effects, we report two-panel summaries for data re-uploading and single-upload models. In each figure, the left panel plots adversarial loss versus shots per dimension $s$, and the right panel plots post-attack accuracy versus $s$, comparing PSR with FD at several step sizes $h$.
Figure~\ref{fig:estimator_reupload} covers the data re-uploading case; Figure~\ref{fig:estimator_single} covers the single-upload case.

\begin{figure}[t]
    \centering
    \includegraphics[width=\textwidth]{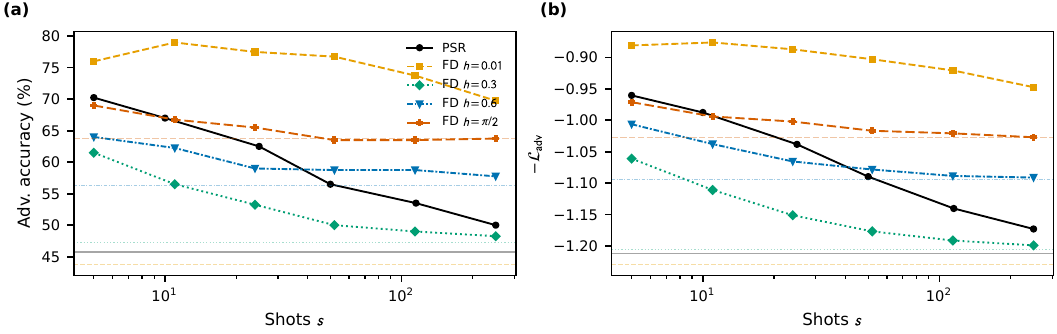}
    \caption{\textbf{Estimator comparison under data re-uploading ($q=6$).} (\textbf{a}) Adversarial accuracy (\%) and (\textbf{b}) negative adversarial loss $-L_{\mathrm{adv}}$, both versus shots $s$, for PSR and FD at several step sizes $h$ under matched total shot budgets.}
    \label{fig:estimator_reupload}
\end{figure}

\begin{figure}[t]
    \centering
    \includegraphics[width=\textwidth]{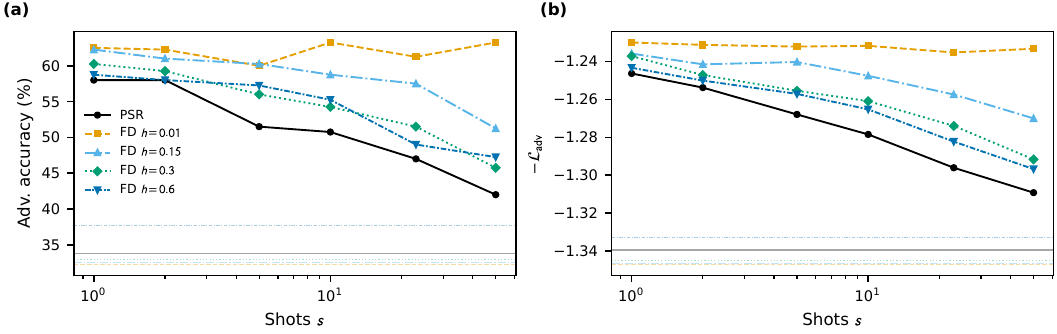}
    \caption{\textbf{Estimator comparison under single upload ($q=6$).} (\textbf{a}) Adversarial accuracy (\%) and (\textbf{b}) negative adversarial loss $-L_{\mathrm{adv}}$, both versus shots $s$, for PSR and FD at several step sizes $h$ under matched total shot budgets.}
    \label{fig:estimator_single}
\end{figure}

Across the shared $(s,h)$ sweep, the re-uploading architecture (Fig.~\ref{fig:estimator_reupload}) exhibits a small low-shot crossover: FD's reduced variance can outweigh its $O(h^2)$ bias when $s$ is tiny, after which PSR's unbiased gradients dominate as $\kappa_{\mathrm{vMF}}\propto s\|\gstar\|^2/\sigma^2$ increases per (\ref{eq:CosSim}). For the single-upload variant (Fig.~\ref{fig:estimator_single}), PSR leads for all $s$; here, the FD surrogate provides no advantage because the analytic PSR effectively acts as a large-step finite difference without incurring bias. Both behaviors are consistent with Proposition~\ref{prop:fgsm-optimization} and with the covariance structure in (\ref{eq:variance1}), which predicts that reducing estimator bias (PSR) is the dominant lever once shot noise shrinks below the true gradient magnitude.

Taken together, the results indicate that the choice of estimator primarily affects constants and low-shot transients. Under moderate and large shot budgets, PSR's unbiasedness yields superior alignment with the true gradient and better attack efficacy at fixed $R=ds$. In the limit of a large number of shots and low variance (Assumption~\ref{ass:1s_3}), the attacker will prefer the unbiased PSR method. Thus, the inclusion of FD as a gradient estimator does not affect the overall scaling behavior of the attack, and the fundamental shot-dimension law holds, as at least a quadratically increasing amount of shot resources, $R=ds$, is required to maintain constant adversarial performance as the data/model dimension $d$ increases.

\paragraph{Takeaway.} Estimator choice changes constants and low-shot crossovers but not the core scaling: PSR dominates at moderate/large $s$ because it satisfies the unbiased, shot-limited model assumed in Assumptions~\ref{ass:1s_2}--\ref{ass:1s_3}, while FD can transiently help at tiny $s$ via bias-variance trade-offs that momentarily shrink the covariance in (\ref{eq:variance1}).

\section{Zero-order gradient estimation under finite shot budgets}
\label{app:zero-order}

The bounds of Section~\ref{sec:singlestep} assume an unbiased, first-order estimator, the parameter-shift rule. A natural objection is that a zero-order estimator such as SPSA (simultaneous-perturbation stochastic approximation) sidesteps the $\Theta(d)$ circuit count by probing the loss along one random direction, using two circuit evaluations regardless of $d$. Classically this saving is real: zero-order estimators built from random directions cost only a $\sqrt{d}$ factor over first-order methods~\cite{nesterov_random_2017, duchi_optimal_2015}, with matching information-theoretic lower bounds under fixed-smoothness assumptions~\cite{alabdulkareem_lower_2021}. Whether the saving survives the finite-shot, high-curvature regime of expressive quantum circuits is the question this appendix frames.

\emph{Notation.} Throughout this appendix $s$ denotes shots per circuit evaluation, matching the per-evaluation count $s_0$ of Section~\ref{sec:notation}; the paper's per-dimension convention follows from $s_{\text{paper}} = 2s$ for commuting-observable PSR.

The SPSA gradient estimate~\cite{spall_multivariate_1992, hoffmann_gradient_2022} is
\begin{equation}
\hat{\nabla}_{\text{SPSA}} f(\theta) = \frac{f(\theta + c\Delta) - f(\theta - c\Delta)}{2c}\,\Delta\;, \qquad \Delta \sim \text{Rademacher}^d\;,
\end{equation}
with finite-difference step $c>0$. It uses two evaluations independent of $d$, but each is a quantum expectation value read from finitely many shots, so $\hat{f}(\theta) \sim \mathcal{N}(f(\theta), \sigma^2/s)$ and the shot noise survives; the random projection $\Delta$ moves it into the estimator's variance and spreads it across all $d$ coordinates,
\begin{equation}
\Var[\hat{\nabla}_{\text{SPSA}}] = \E_{\Delta}\!\left[\frac{d\sigma^2}{c^2 s} + O(c^2)\right] = O\!\left(\frac{d\sigma^2}{c^2 s}\right) + \text{bias}^2\;,
\end{equation}
which still scales as $d/s$. The finite-difference step also carries a bias $\E[\hat{\nabla}_{\text{SPSA}}] - \nabla f = O(c^2\|\nabla^2 f\|)$~\cite{spall_multivariate_1992}. Table~\ref{tab:psr-spsa} collects the trade-off.

\begin{table}[t]
\centering
\begin{tabular}{lccc}
\hline
\textbf{Method} & \textbf{Circuit Evals} & \textbf{Gradient Variance} & \textbf{Bias} \\
\hline
PSR & $2d$ & $\Theta(\sigma^2/s)$ per component & 0 \\
SPSA & $2$ & $\Theta(d\sigma^2/(c^2 s))$ total & $O(c^2)$ \\
\hline
\end{tabular}
\caption{\textbf{PSR vs.\ SPSA under finite shots.} Comparison of circuit count, gradient variance, and bias.}
\label{tab:psr-spsa}
\end{table}

The two knobs pull against each other: bias forces $c$ small and variance forces $s$ large, so reaching a fixed gradient accuracy costs $s = \Theta(d/c^2)$ shots and the total budget climbs back toward PSR's $\Theta(ds)$. In a low-curvature loss the bias term is mild and the trade can still net the classical $\sqrt{d}$ advantage~\cite{duchi_optimal_2015}; in the high-curvature regime characteristic of expressive circuits the bias grows with $\|\nabla^2 f\|$ and that advantage is no longer guaranteed. The only adversarial-QML study to use hybrid (zero-order) gradients, by Majumder et al.~\cite{Majumder2021Hybrid}, reports feasibility on a single problem instance without shot-budget or dimension-scaling analysis, so the question stays empirically open. Our scaling law fixes the exponent for the unbiased regime (Section~\ref{sec:singlestep}); whether a biased zero-order estimator can lower it for expressive, classically-hard models, and how the optimal step $c^*(d,s)$ trades bias against variance, is left to future work.

\newpage
\sloppy
\bibliographystyle{iopart-num}
\bibliography{Pershot2, previous_ref, ref_added}

\end{document}